%% file: main.tex
\documentclass{article}
\usepackage{ijcai16}
\usepackage{times}
\usepackage{amsmath}
\usepackage{amsthm}
\usepackage{amssymb}
\usepackage{thm-restate}

\usepackage[T1]{fontenc}

\PassOptionsToPackage{table}{xcolor}
\usepackage{graphics}
\usepackage{tikz}
\usepackage{url}
\usepackage{enumerate}

\usepackage{booktabs}
\usepackage{colortbl}
\usepackage{wrapfig}

\usepackage{times}

\usepackage{xspace}
\usetikzlibrary{arrows,backgrounds,positioning}
\usetikzlibrary{snakes}

\def\calO{\mathcal{O}\xspace}

\newcommand{\NE}{Nash equilibrium\xspace}
\newcommand{\NEs}{Nash equilibria\xspace}

\newcommand{\lasts}{\mathit{last}}

\input{macros.tex}

\title{Efficient Local Search in Coordination Games on Graphs}
\author{Sunil Simon \\IIT Kanpur\\Kanpur, India\\
\And Dominik Wojtczak \\University of Liverpool\\
Liverpool, U.K. \\
}

\begin{document}
\maketitle

\begin{abstract}
We study strategic games on weighted directed graphs, where the payoff
of a player is defined as the sum of the weights on the edges from
players who chose the same strategy 
augmented by a fixed non-negative bonus for picking 
a given strategy.  These games capture the idea of
coordination in the absence of globally common strategies. Prior work
shows that the problem of determining the existence of a pure Nash
equilibrium for these games is NP-complete already for graphs with all
weights equal to one and no bonuses.  However, for several classes of graphs
(e.g. DAGs and cliques) pure Nash equilibria or even strong equilibria
always exist and can be found
by simply following a particular improvement or coalition-improvement
path, respectively.  In this paper we identify several natural classes
of graphs for which a finite improvement or coalition-improvement path
of polynomial length always exists, and, as a consequence, a \NE or
strong equilibrium in them can be found in polynomial time.  We also
argue that these results are optimal in the sense that in natural
generalisations of these classes of graphs, a pure \NE may not even exist.
\end{abstract}

\input{intro.tex}
\input{prelim.tex}

\input{model.tex}

\input{simple-cycle.tex}
\input{necklace.tex}

\input{partition.tex}
\input{no-way-to-se.tex}
\bibliographystyle{abbrv}
\bibliography{e,clustering,extrabib}

\input{appendix}
\end{document}

%% file: macros.tex
\newtheorem{theorem}{Theorem}

\newtheorem{proposition}[theorem]{Proposition}

\newtheorem{lemma}[theorem]{Lemma}
\newtheorem{example}[theorem]{Example}
\newtheorem{corollary}[theorem]{Corollary}

\newcommand{\betredge}[1]{\betstep{#1}}
\newcommand{\betstep}[1]{\mbox{$\stackrel{#1}{\rightarrow}$}}

\newcommand{\guard}{\mathit{guard}}
\newcommand{\ochain}{open chain of cycles}
\newcommand{\cchain}{closed chain of cycles}
\def\ochains{open chains of cycles\xspace}
\def\cchains{closed chains of cycles\xspace}

\newcounter{symbol}
\newcommand{\indexsyma}[1]%
{\stepcounter{symbol}\index{zzz1 \thesymbol @\protect#1}}
\newcommand{\indexsymb}[1]%
{\stepcounter{symbol}\index{zzz2 \thesymbol @\protect#1}}
\newcommand{\indexsymc}[1]%
{\stepcounter{symbol}\index{zzz3 \thesymbol @\protect#1}}
\newcommand{\indexsymd}[1]%
{\stepcounter{symbol}\index{zzz4 \thesymbol @\protect#1}}
\newcommand{\indexsyme}[1]%
{\stepcounter{symbol}\index{zzz5 \thesymbol @\protect#1}}

\newcommand{\bfe}[1]{\begin{bfseries}\emph{#1}\end{bfseries}\index{#1}}

\newcommand{\myra}{\mbox{$\:\rightarrow\:$}}

\newcommand{\LL}{\mbox{$\ldots$}}

\newcommand{\C}[1]{\mbox{$\{{#1}\}$}}           %

\newcommand{\szkew}[1]{\relax \setbox0=\hbox{\kern -24pt $\displaystyle#1$\kern 0pt }%
\box0}
{\catcode`\@=11 \global\let\ifjusthvtest@=\iffalse}
\newcounter{oldmycaption}

%% file: intro.tex
\section{Introduction}
\label{sec:intro}

Nash equilibrium is an important solution concept in game theory which
has been widely used to reason about strategic interaction between
rational agents. Although Nash's theorem guarantees existence of a mixed
strategy Nash equilibrium for all finite games, pure strategy Nash
equilibria need not always exist.
In many scenarios of strategic interaction, apart from the question of
existence of pure Nash equilibria, an important concern is whether it
is possible to compute an equilibrium outcome 
and whether a game always converges to one.
The concept
of {\it improvement paths} is therefore fundamental in the study of
strategic games. Improvement paths are essentially maximal paths
constructed by starting at an arbitrary joint strategy and allowing
players to improve their choice in a unilateral manner. At each stage,
a single player who did not select a best response is allowed to
update his choice to a better strategy. By definition, every finite
improvement path terminates in a Nash equilibrium. In a seminal
paper, Monderer and Shapley \shortcite{MS96} studied the class of games in
which every improvement path is guaranteed to be finite, which was coined as
the {\it finite improvement property} (FIP). They
showed that games with the FIP are precisely those games to which we
can associate a generalised ordinal potential function.
Thus FIP not only guarantees the existence of pure Nash equilibria but also
ensures that it is possible to converge to an equilibrium outcome by
performing {\it local search}.
This makes FIP a desirable property to have in any game. An important
class of games that have the FIP is {\it congestion
  games} \cite{Ros73}.
However, the requirement that {\it every} improvement path is finite,
turns out to be a rather strong condition and there are very
restricted classes of games that have this property.

Young \shortcite{You93} proposed weakening the finite improvement property
to ensure the {\it existence} of a finite improvement path starting
from any initial joint strategy. Games for which this property hold
are called {\it weakly acyclic games}. Thus weakly acyclic games
capture the possibility of reaching pure Nash equilibria through
unilateral deviations of players irrespective of the starting
state. Milchtaich \shortcite{Mil96} showed that although congestion games
with player specific payoff functions do not have the FIP, they are
weakly acyclic. Weak acyclicity of a game also ensures that certain
modifications of the traditional no-regret algorithm yields almost
sure convergence to a pure Nash equilibrium \cite{MardenAS07}.

Although finite improvement path guarantees the existence of a Nash
equilibrium, it does not necessarily provide an efficient algorithm to
compute an equilibrium outcome. In many situations, improvement paths
could be exponentially long. In fact, Fabrikant {\em et al.~}\shortcite{FPT04}
showed that computing a pure Nash equilibrium in congestion games is PLS-complete.
Even for symmetric network congestion games, where it is
known that a pure Nash equilibrium can be efficiently computed
\cite{FPT04}, there are classes of instances where any best response
improvement path is exponentially long \cite{ARV06}. Thus identifying
classes of games that have finite improvement paths in which it is
possible to converge to a Nash equilibrium in a polynomial number of
steps is of obvious interest.
	
In game theory, coordination games are often used to model situations
in which players attain maximum payoff when agreeing on a common
strategy. In this paper, we study a simple class of coordination games
in which players try to coordinate within a certain neighbourhood.
The neighbourhood structure is specified by a finite directed graph
whose nodes correspond to the players. Each player chooses a colour
from a set of available colours. 
The payoff of a player is defined as the sum of the
weights on the edges from players who choose the same colour 
and a fixed bonus for picking that particular colour.
These games
constitute a natural class of strategic games, which capture the
following three key characteristics. {\it Join the crowd property:} the payoff of each player weakly
  increases when more players choose her strategy. {\it Asymmetric strategy sets:} players may have different strategy sets. {\it Local dependency:} the payoff of each player depends only
  on the choices made by the players in its neighbourhood. 

A similar model of coordination games on graphs was introduced in
\cite{ARSS14} where the authors considered undirected graphs. 
However, the transition from undirected to directed graphs drastically changes the
status of the games. For instance, in the case of undirected graphs,
coordination games have the FIP. While in the directed case, Nash
equilibria may not exist. Moreover, the problem of determining
the existence of Nash equilibria is NP-complete for coordination games
on directed graphs \cite{ASW15}. However, if the underlying graph is a
directed acyclic graph (DAG), a complete graph or a simple cycle, then
pure Nash equilibria always exist.
These proofs %
can easily be adapted to show that weighted DAGs and
weighted simple cycles have finite improvement paths.

\smallskip
\noindent{\bf Related work.} Although the class of potential games are
well studied and has been a topic of extensive research, weakly
acyclic games have received less attention. Engelberg and Schapira
\shortcite{ES11} showed that certain Internet routing games are weakly
acyclic. In a recent paper Kawald and Lenzner \shortcite{KL13} show that
certain classes of network creation games are weakly acyclic and
moreover that a specific scheduling of players can ensure that the
resulting improvement path converges to a Nash equilibrium in
$\mathcal{O}(n \log n)$ steps. %
Brokkelkamp and Vries \shortcite{BV12} improved Milchtaich's result
\shortcite{Mil96} on congestion games with player specific payoff functions
by showing that a specific scheduling of players is sufficient to
construct an improvement path that converges to a Nash equilibrium. 

Unlike in the case of exact potential games, there is no neat
structural characterisation of weakly acyclic games. Some attempts in
this direction has been made in the past.  Fabrikant {\em et
al.~}\shortcite{FabrikantJS10} proved that the existence of a unique (pure)
Nash equilibrium in every sub-game implies that the game is weakly
acyclic. A comprehensive classification of weakly acyclic games in
terms of schedulers is done in \cite{AS12}. Finally, Milchtaich
\shortcite{Mil13} showed that every finite extensive-form game with perfect
information is weakly acyclic.

The model of coordination games are related to various well-studied
classes of games. Coordination games on graphs are {\it polymatrix
  games} \cite{Jan68}. In these games, the payoff for each player is
the sum of the payoffs from the individual two player games he plays
with every other player separately. Hoefer \shortcite{Hoefer2007} studied
clustering games that are also polymatrix games based on undirected
graphs. However, in this setup each player has the same set of
strategies and it can be shown that these games have the FIP. 
A model that does not assume all strategies to be the same, 
but is still based on undirected graphs,
was shown to have the FIP in \cite{RS15}.
When the graph is
undirected and complete, coordination games on graphs are special
cases of the monotone increasing congestion games that were studied in
\cite{RT06}.  
\noindent  
\begin{table}[t]
\newcolumntype{x}[1]{>{\centering\let\newline\\\arraybackslash\hspace{0pt}}p{#1}}
\setlength{\tabcolsep}{0.01cm}
\renewcommand{\arraystretch}{1.2}
\rowcolors{2}{white}{gray!20}
\resizebox{1.01\columnwidth}{!}{
\newlength{\myl}
\settowidth{\myl}{weighted simple cycles+2 b}
\begin{tabular}{x{\myl}cc}
\noalign{\global\belowrulesep=0.0ex}
\toprule
\noalign{\global\aboverulesep=0.0ex}
Graph Class & improvement path & c-improvement path \\ 
\midrule
weighted DAGs & $\mathcal{O}(n)$ \cite{ASW15} & $\mathcal{O}(n)$ \cite{ASW15} \\
weighted simple cycles with 2 bonuses & $\mathcal{O}(n)$ [Thm. \ref{thm:cycle-2bonuses}] & $\mathcal{O}(n)$ [Thm. \ref{thm:se-cycle}] \\
open chains of cycles  & $\mathcal{O}(nm^2)$ [Thm. \ref{thm:necklace-noweight-nobonus}] & $\mathcal{O}(nm^3)$ [Cor. \ref{cor:se-open-chain}] \\
closed chains of cycles  & $\mathcal{O}(nm^2)$ [Thm. \ref{thm:necklace-cycle-noweights}] & $\mathcal{O}(nm^3)$ [Thm. \ref{thm:se-open-chain}] \\
weighted open chains of cycles  & $\mathcal{O}(nm^3)$ [Thm. \ref{thm:necklace-nobonus}] & ?? \\
weighted closed chains of cycles  & \multicolumn{2}{c}{\NE may not exist [Ex. \ref{ex:necklace-weight}]} \\
partition-cycles  & $\mathcal{O}(n(n-k))$ [Thm. \ref{thm:partcycle-noweight-nobonus}] & ?? \\
partition-cycles+bonuses & $\mathcal{O}(kn(n-k))$ [Thm. \ref{thm:partcycle-noweight-bonus}] & ?? \\
weighted partition-cycles  & \multicolumn{2}{>{\columncolor{gray!20}}c}{\NE may not exist [Ex. \ref{ex:partcycle-noNE}]}\\
\noalign{\global\aboverulesep=0.0ex}
\bottomrule
\end{tabular}
}
\vskip0.5em
\caption{
\label{fig:summary}
An upper bound on the length of the shortest improvement and c-improvement path
for a given class of graphs. 
All edges are unweighted and there are no bonuses 
unless the name of the class says otherwise.
For simple cycles
and chains of cycles we assume that each cycle has $n$ nodes and the
number of cycles in the chain is $m$.  For partition-cycles, $n$ is
the total number of nodes and $1 \leq k < n$ is the number of nodes in
the top part of the cycle (set $V_T$).
}
\end{table}

\smallskip
\noindent {\bf Our contributions.}  In this paper, we identify some
natural classes of polymatrix games based on the coordination game
model, which even though do not have the FIP (cf. Example 4 in \cite{ASW15}), are weakly acyclic. We
also show that for these games a finite improvement path of polynomial
length can be constructed in a uniform manner. Thus not only do these
games have pure Nash equilibria, but they can also be efficiently
computed by local search.

We start by analysing coordination games on simple cycles. Even in
this simple setting, improvement paths of infinite length may
exist. However, we show that there always exists a finite improvement
path of polynomial length. 
We then extend the setting of simple cycles in two directions. First
we consider chains of simple cycles where we show that polynomial
length improvement paths exist. We then consider simple cycles with
cross-edges and show the existence of polynomial length improvement
paths. We also demonstrate that these results are optimal in the sense
that most natural generalisations of these structures may result in
games in which a Nash equilibrium may not even exists. Most of our
constructions involve a common proof technique: we identify a specific
scheduling of players using which, starting at an arbitrary initial
joint strategy, we can reach a joint strategy in which at most two
players are not playing their best response.
We argue that such a joint strategy can then be updated to converge to
a Nash equilibrium. We also identify a structural condition on
coalitional deviation once a Nash equilibrium is attained. This property
is then used to show the existence of a finite ``coalitional''
improvement path which terminates in a strong equilibrium.
Our results also imply an almost sure convergence,
although not necessarily in a polynomial number of steps,
to a Nash equilibrium when the order of deviations is random, but ``fair''.
Fairness requires that
for any deviation,
there 
is a fixed nonzero lower bound on 
the probability of it taking place
from any state of the game where it can be taken.
Note that this implies that the same holds for any finite sequence of deviations.
A Nash equilibrium is reached almost surely with such a random order of deviations, 
because when starting at any state we either 
follow a finite improvement path to a Nash equilibrium 
with a nonzero probability or 
that path stops in some new state from 
where we can follow another finite improvement path
with a nonzero probability.
As this repeats over and over again, 
almost surely one such finite improvement path will succeed.

Table \ref{fig:summary} summarises most of our results. 

\smallskip
\noindent {\bf Potential applications.} 
Coordination games constitute an abstract game model which is well
studied in game theory and has been shown to model many practical
scenarios. The game model that we consider in this paper is an
extension of coordination games to the network setting (in which the
neighbourhood relation is specified using a directed graph) where
common strategies are not guaranteed to exist and payoffs are not
necessarily symmetric.

The graph classes that we consider are typical for network topologies,
e.g. token ring local area networks are organised in directed simple
cycles, open chains topology is supported by recommendation G.8032v2 on
Ethernet ring protection switching, and closed chains are used in
multi-ring
protocols.
The basic technique that we use to show convergence to Nash equilibria
is based on finite improvement paths of polynomial length. The concept
of an improvement path is fundamental in the study of games and it has
been used to explain and analyse various real world applications. One
such example is the border gateway protocol (BGP) which establishes
routes between competing networks on the Internet. Over the years,
there has been extensive research, especially in network
communications literature, on how stable routing states are achieved
and maintained in BGP in spite of strategic concerns. Fabrikant and
Papadimitriou \cite{FP08} and independently, Levin and
others \cite{LSZ08} observed that BGP can be viewed as best-response
dynamics in a class of routing games and finite improvement paths that
terminates in a pure Nash equilibria essentially translates to stable
routing states. Following this observation, Engelberg and
Schapira \cite{ES11} presents a game theoretic analysis of routing on
the Internet where they show weak acyclicity of various routing games.

The coalition formation property inherent to coordination games on
graphs also make the game model relevant to cluster analysis. In
cluster analysis, the task is to organise a set of objects into groups
according to some similarity measure. Here, the strategies can be
viewed as possible cluster names and a pure NE naturally corresponds
to a `satisfactory' clustering of the underlying graph.  Clustering
from a game theoretic perspective was for instance applied to car and
pedestrian detection in images, and face recognition
in \cite{PB14}. This approach was shown to perform very well against
the state of the art.

\smallskip

\noindent {\bf Structure of the paper.}  In
Section~\ref{sec:colouring} we introduce the game model and make an
important observation on the structure of coalition deviation from a
\NE in coordination games on directed graphs. In
Section~\ref{sec:simple-cycle} we analyse games whose underlying
graphs are simple cycles. In Section~\ref{sec:necklace} we study
chains of cycles and in Section \ref{sec:partion-cycles} we consider
simple cycles with cross edges. 

%% file: prelim.tex
\section{Preliminaries}
\label{sec:prelim}

A \bfe{strategic game} $\mathcal{G}=(S_1, \ldots, S_n,$ $p_1, \ldots,
p_n)$ with $n > 1$ players, consists of a non-empty set $S_i$ of
\bfe{strategies} and a \bfe{payoff function} $p_i : S_1 \times \cdots
\times S_n \myra \mathbb{R}$, for each player $i$.  We denote $S_1
\times \cdots \times S_n$ by $S$, call each element $s \in S$ a
\bfe{joint strategy} and abbreviate the sequence $(s_{j})_{j \neq i}$
to $s_{-i}$. Occasionally we write $(s_i, s_{-i})$ instead of $s$.  We
call a strategy $s_i$ of player $i$ a \bfe{best response} to a joint
strategy $s_{-i}$ of his opponents if for all $ s'_i \in S_i$,
$p_i(s_i, s_{-i}) \geq p_i(s'_i, s_{-i})$.

We call a non-empty subset $K := \{k_1, \ldots, k_m\}$ of the set of
players $N:= \{1, \ldots, n\}$ a \bfe{coalition}. Given a joint
strategy $s$ we abbreviate the sequence $(s_{k_1}, \ldots, s_{k_m})$
of strategies to $s_K$ and $S_{k_1} \times \cdots \times S_{k_m}$ to
$S_{K}$. We occasionally write $(s_K, s_{-K})$ instead of $s$. If
there is a strategy $x$ such that $s_i = x$ for all players $i \in K$,
we also write $(x_K, s_{-K})$ instead of $s$.

Given two joint strategies $s'$ and $s$ and a coalition $K$, we say
that $s'$ is a \bfe{deviation of the players in $K$} from $s$ if $K =
\{i \in N \mid s_i \neq s_i'\}$.  We denote this by $s \betredge{K}
s'$. If in addition $p_i(s') > p_i(s)$ holds for all $i \in K$, we say
that the deviation $s'$ from $s$ is \bfe{profitable}. Further, we say
that a coalition $K$ \bfe{can profitably deviate from $s$} if there
exists a profitable deviation of the players in $K$ from $s$.  Next,
we call a joint strategy $s$ a \bfe{k-equilibrium}, where $k \in \{1,
\dots, n\}$, if no coalition of at most $k$ players can profitably
deviate from $s$.  Using this definition, a \bfe{Nash equilibrium} is
a 1-equilibrium and a \bfe{strong equilibrium}, see \cite{Aumann59},
is an $n$-equilibrium.

A \bfe{coalitional improvement path}, in short a \bfe{c-improve\-ment
  path}, is a maximal sequence $\rho=(s^1, s^2, \dots)$ of joint
strategies such that for every $k > 1$ there is a coalition $K$ such
that $s^k$ is a profitable deviation of the players in $K$ from
$s^{k-1}$. If $\rho$ is finite then by $\lasts(\rho)$ we denote the last
element of the sequence.  Clearly, if a c-improvement path is finite,
its last element is a strong equilibrium. 
We say that $\mathcal{G}$ is \bfe{c-weakly acyclic} if for every joint
strategy there exists a finite c-improvement path that starts at
it. Note that games that are c-weakly acyclic have a strong
equilibrium.
We call a c-improvement path an \bfe{improvement path} if each
deviating coalition consists of one player. The notion of a game
being \bfe{weakly acyclic} \cite{You93,Mil96}, is then defined by
referring to improvement paths instead of c-improvement paths.

%% file: model.tex
\section{Coordination games on directed graphs}
\label{sec:colouring}

We now define the class of games we are interested in.  Fix a
finite set of colours $M$. A weighted directed graph $(G,w)$ is a
structure where $G=(V,E)$ is a graph without self loops over the
vertices $V=\{1,\ldots,n\}$ and $w$ is a function that associates with
each edge $e \in E$, a non-negative weight $w_e$. We say that a node
$j$ is a \bfe{neighbour} of the node $i$ if there is an edge $j \to i$
in $G$.  Let $N_i$ denote the set of all neighbours of node $i$ in the
graph $G$.  A \bfe{colour assignment} is a function $C: V \to 2^M$ which
assigns to each node of $G$ a finite non-empty set of colours. We also
introduce the concept of a \bfe{bonus}, which is a function $\beta$
that to each node $i$ and a colour $c \in M$ assigns a natural number
$\beta(i,c)$.  Note that bonuses can be modelled by incoming edges
from fixed colour source nodes, i.e. nodes with no incoming edges and
only one colour available to them.
When stating our results, bonuses are assumed to be not present, 
unless we explicitly state that they are allowed. 
Bonuses are extensively used in our proofs because 
a coordination game restricted to a given subgraph 
can be viewed as a coordination game with bonuses 
induced by the remaining nodes of the graph.

Given a weighted graph $(G,w)$, a colour assignment $C$ and a bonus
function $\beta$ a strategic game $\mathcal{G}(G,w,C,\beta)$ is
defined as follows: the players are the nodes,
\begin{itemize}
\item the set of strategies of player (node) $i$ is the set of colours
  $C(i)$; we occasionally refer to the strategies as \bfe{colours}.

\item the payoff function $p_i(s) = \sum_{j \in N_i,\, s_i = s_j} w_{j
  \to i} + \beta(i,s_i)$.
\end{itemize}

So each node simultaneously chooses a colour and the payoff to the
node is the sum of the weights of the edges from its neighbours that
chose its colour augmented by the bonus to the node from choosing the
colour.  We call these games \bfe{coordination games on directed
  graphs}, from now on just \bfe{coordination games}.  
When the
weights of all the edges are 1, we obtain a coordination game whose
underlying graph is unweighted. In this case, we simply drop the
function $w$ from the description of the game. Similarly if all the
bonuses are 0 then we obtain a coordination game without
bonuses. Likewise, to denote this game we omit the function
$\beta$. In a coordination game without bonuses where the underlying
graph is unweighted, each payoff function is defined by $p_i(s) :=
|\{j \in N_i \mid s_i = s_j\}|$.

\begin{figure}[htbp]
\begin{minipage}{0.5\textwidth}
\normalsize
\begin{example}{(\cite{ASW15})}
\label{exa:payoff}
Consider the directed graph and the colour assignment
depicted in Figure~\ref{fig:example-graph}.
Take the joint strategy $s$ that consists of the underlined strategies.
Then the payoffs are as follows:
\begin{itemize}
\item 0 for the nodes 1, 7, 8 and 9,
\item 1 for the nodes 2, 4, 5, 6, 
\item 2 for the node 3.
\end{itemize}
Note that the above joint strategy is not a Nash equilibrium. For example,
node 1 can profitably deviate to colour $a$.
\qed
\end{example}
\end{minipage} \hfill
\begin{minipage}{0.5\textwidth}
\centering
\tikzstyle{agent}=[circle,draw=black!80,thick, minimum size=2em,scale=0.8]
\begin{tikzpicture}[auto,>=latex',shorten >=1pt,on grid]
\newdimen\R
\R=1.3cm
\newcommand{\llab}[1]{{\small $\{#1\}$}}
\draw (90: \R) node[agent,label=right:{\llab{a,\underline{b}}}] (1) {1};
\draw (90-120: \R) node[agent,label=right:{\llab{a,\underline{c}}}] (2) {2};
\draw (90-240: \R) node[agent,label=left:{\llab{b,\underline{c}}}] (3) {3};
\draw (30: \R) node[agent,label=right:\llab{a,\underline{b}}] (4) {4};
\draw (30-120: \R) node[agent,label=right:{\llab{a,\underline{c}}}] (5) {5};
\draw (30-240: \R) node[agent,label=left:{\llab{b,\underline{c}}}] (6) {6};
\draw (90: 1.7*\R) node[agent,label=right:{\llab{\underline{a}}}] (7) {7};
\draw (90-120: 2*\R) node[agent,label=right:{\llab{\underline{c}}}] (8) {8};
\draw (90-240: 2*\R) node[agent,label=left:{\llab{\underline{b}}}] (9) {9};
\foreach \x/\y in {1/2,2/3,3/1,1/4,4/2,2/5,5/3,3/6,6/1,7/1,8/2,9/3} {
    \draw[->] (\x) to (\y);    
}
\end{tikzpicture}
\caption{A directed graph with a colour assignment. \label{fig:example-graph}}
\end{minipage} 
\end{figure}
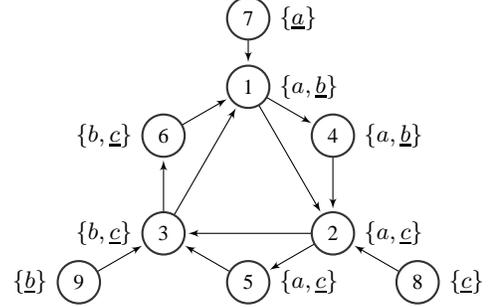

Finally, given a directed graph $G$ and a set of nodes $K$, we denote
by $G[K]$ the subgraph of $G$ induced by $K$.

We now show a structural property of a coalition deviation from a \NE
in our coordination games.  This will be used later to prove c-weak
acyclicity for a class of games based on their weak acyclicity. Note
that this cannot be done for all classes of graphs, because there
exist a coordination game on undirected graph which is weakly acyclic,
but has no strong equilibrium \cite{ARSS14}.

\begin{lemma}
\label{lem:unicolor-cycle}
Any profitable coalition deviation from a \NE includes a unicoloured
directed simple cycle.
\end{lemma}
\begin{proof}
Let $s$ be any \NE in the game and 
let coalition $K$ have a profitable deviation, $s'$,
from $s$.
It suffices to show that each node in $K$ has a predecessor in $K$ deviating to the same colour.
Assume that for some player $i\in K$ it is not the case.
We then have the following:
$p_i(s) < p_i(s'_K,s_{-K}) =$ 
$\sum_{j\in K: s'_{j} = s'_i} w_{j\to i}$ $+ \sum_{j\not\in K: s_{j} = s'_i} w_{j\to i}$ $+ \beta(i,s'_i) 
\leq 0 + \sum_{j: s_{j} = s'_i} w_{j\to i} + \beta(i,s'_i) = p_i(s'_i,s_{-i})$,
so player $i$ would also be able to improve his payoff by unilaterally switching to $s'_i$ in $s$, which 
contradicts the fact that $s$ is a \NE.%
\end{proof}

%% file: simple-cycle.tex
\section{Simple cycles}
\label{sec:simple-cycle}

In this section we focus on the case when the game graph is a directed simple cycle. 
Despite the simplicity of this model the problems we consider are already nontrivial 
for such a basic graph structure. 
We first restate a result from \cite{ASW15} where unweighted graphs
are considered. To fix the notation, suppose that
the considered graph is $1 \to 2 \to \ldots \to n \to 1$.  Below for
$i \in \C{2,\LL,n}$, $i \ominus 1=i-1$, and $1 \ominus 1=n$.

\begin{theorem}{(\cite{ASW15})}
\label{thm:TARK15}
Every coordination game with bonuses on an unweighted simple cycle 
has a c-improvement path of length $\mathcal{O}(n)$.
\end{theorem}

We would like to extend this result to weighted graphs with bonuses.
However as the following example demonstrates, if in a simple cycle,
we allow non-trivial weights on at least three edges and associate
bonuses with at least three nodes then there are coordination games
that need not even have a Nash equilibrium.

\begin{example}
\label{ex:noNE}
Consider the simple cycle on three nodes 1, 2 and 3 in which all the
edges have weight 2. Let $C(1)=\{a,b\}$, $C(2)=\{a,c\}$ and
$C(3)=\{b,c\}$. Let the bonus be defined as 
$\beta(1,a) = \beta(2,c) = \beta(3,b) = 1$ 
and equal to $0$ otherwise.
The structure essentially corresponds
to the one shown in Figure~\ref{fig:example-graph}.  The resulting
coordination game does not have a Nash equilibrium. Below we list all
the joint strategies and we underline a strategy that is not a best
response to the choice of other players: $(\underline{a},a,b)$,
$(a,a,\underline{c})$, $(a,c,\underline{b})$, $(a,\underline{c},c)$,
$(b,\underline{a},b)$, $(\underline{b},a,c)$, $(b,c,\underline{b})$
and $(\underline{b},c,c)$.
\qed
\end{example}

We show here that this counterexample is essentially minimal, i.e.  if
only two nodes have bonuses or only two edges have weights then the
coordination game is weakly acyclic. 

\begin{restatable}{theorem}{cycletbonus}
\label{thm:cycle-2bonuses}
Every coordination game on a weighted simple cycle in which at most
two nodes have bonuses 
has an improvement path of length $\mathcal{O}(n)$.
\end{restatable}

\begin{proof}
Assume without loss of generality that one of the nodes which has a
bonus is node $1$ (otherwise we can re-label the nodes on the
cycle). Let the other vertex with a bonus be some $k \in N$. Let $s$
be an arbitrary joint strategy. We perform the following sequence of best
response updates.

We proceed around the cycle in the order $1, \ldots, n$ and let
players switch to any of their best responses. We argue that in at most three
rounds, the resulting improvement path terminates in a Nash
equilibrium. At the end of the first round, players $2,\ldots,n$ are
playing their best response. If the resulting joint strategy $s^1$ is
a Nash equilibrium, then we stop. Otherwise player 1 strategy $s^1_1$ is
not a best response to $s^1_{-1}$. Let player $1$ update his strategy,
denote the resulting joint strategy $s^2$. There are two cases:

\begin{itemize}
\item Suppose $s^2_1=s^1_n$. We proceed around the cycle in the cyclic
  order up until the node $k-1$ and update the strategy of each
  player. Note that if at some point in between we reach a Nash
  equilibrium then we stop, otherwise the only colour that is
  propagated along the cycle until node $k-1$ is $s^1_n$. Let the
  resulting joint strategy be $s^3$. Now suppose $s^3_k$ is not a best
  response to $s^3_{-k}$. Let player $k$ update his strategy and call
  the resulting joint strategy $s^4$. If $s^4_k = s^3_{k-1}
  (=s^1_{n})$ then we continue around the cycle making players update to
  their best response. This improvement path is guaranteed to
  terminate since the only colour which is propagated is $s^1_{n}$. If
  $s^4_k \neq s^3_{k-1}$, then $s^4_k = c^k$ for some $c^k \in
  C(k)$. Continue in the cyclic order from $k+1, \ldots, n$ making
  players update to their best response. Let the resulting joint
  strategy be $s^5$. Note that in this sequence if a player switches
  then it is to the colour $s^4_k=c^k$. 

Suppose $s^5_1$ is not a best response to $s^5_{-1}$, then let player
1 update and call the resulting joint strategy $s^6$. If
$s^6_1=s^6_{n} (=c^k)$ then continue in the cyclic order from $2,
\ldots, k-1$. The only colour which is propagated is $c^k$ and this
improvement path is finite since the colour chosen by $k$ is $c^k$. If
$s^6_1\neq s^6_{n}$ then $s^6_1 = c^1$ for some $c^1 \in C(1)$. 

Now let players update to their best response in the cyclic order $2,
\ldots, n$. Either the improvement path terminates before player $k$
updates since his best response remains $c^k$ or player $k$ updates to
$c^1$ and then the improvement path also terminates since the only
colour which is propagated in the cycle is $c^1$.

\item Suppose $s^2_1 \neq s^1_{n}$, then $s^2_1 = c^1$. Proceed around
  the cycle in the cyclic order and let players update to their best
  responses. If player $k$ switches to $c^1$ then the only colour
  which is propagated is $c^1$ and the improvement path terminates in
  one round. Otherwise, player $k$ eventually updates to $c^k$. As in the
  earlier case, let players $k+1, \ldots,n, 1, \ldots, k-1$ update to
  their best response in that order. The resulting improvement path is
  finite.
\end{itemize}
\end{proof}

\noindent This proof can easily be adapted to show the same result for
graphs with at most two weighted edges.

\begin{theorem}
\label{thm:cycle-2weights}
Every coordination game on a simple cycle with bonuses where at
most two edges have non-trivial weights (i.e.\ weights greater than 1) 
has an improvement path of length $\mathcal{O}(n)$.
\end{theorem}

The above results are optimal due to Example~\ref{ex:noNE}.  We can also
show that if a game played on a simple cycle is weakly acyclic, then it
is c-weakly acyclic.

\begin{restatable}{theorem}{secycle}
\label{thm:se-cycle}
In a coordination game played on a weighted simple cycle with bonuses,
any finite improvement path can be extended to a finite c-improvement
path just by adding one profitable coalition deviation step at the end
of it.
\end{restatable}

\begin{proof}
Let us denote by $s$ a \NE that this game reaches via
some finite improvement path.
If $s$ is a strong equilibrium then we are done.  Otherwise there
exists a coalition $K$ with a profitable deviation, $s'$, from $s$.
Due to Lemma \ref{lem:unicolor-cycle}, the coalition $K$ has to
include all players, because there is only one cycle in the game
graph, and all of them have to switch to the same colour in $s'$.  We
argue that $(s',s_{-K}) = s'$ is a \NE. Suppose there is a player $i$,
that can switch to colour $x$ and improve his payoff. Then, $s'_{i
  \ominus 1} \neq x$, because all players play the same colour in
$s'$.  We have $p_i(s) < p_i(s') < p_i((x,s'_{-i})) = \beta(i,x) \leq
p_i((x,s_{-i}))$; a contradiction with the assumption that $s$ is a
\NE.

Finally, let $\rho = s, s^1, s^2, \ldots$ be any c-improvement path.
Due to the above observations every $s^i$ is a \NE where all players play the same colour.
Note that it cannot be $s^i = s^j$ for any $i \neq j$, because 
every $s^{i+1}$ is a profitable deviation from $s^i$.
Therefore any c-improvement path starting at a \NE is finite and its length is at most equal to the number of colours in the game.
However, we can cut this path short by choosing as the first coalition deviation step 
the last colouring in $\rho$.
This would still be a profitable deviation for all the players, because for all $i$ we have
$p_i(s) < p_i(s^1) < p_i(s^2) < \ldots$.
\end{proof}

\begin{corollary}
Every coordination game on a weighted simple cycle in which at most
two nodes have bonuses (or with bonuses but in which at
most two edges have non-trivial weights)
has a c-improvement path of length $\mathcal{O}(n)$.
\end{corollary}

%% file: necklace.tex
\section{Sequence of simple cycles}
\label{sec:necklace}

Next we look at the graph structure which consists of a chain of $m
\geq 2$ simple cycles. Formally, for $j \in \{1,2, \ldots, m\}$, let
$\mathcal{C}_j$ be the cycle $1^j \to 2^j \ldots \to n^j \to 1^j$. For
simplicity, we assume that all the cycles have the same number of
nodes. The results that we show hold for arbitrary cycles as long as
each cycle has at least 3 nodes. An \ochain{}, $\mathcal{N}$ is the
structure in which for all $j \in \{1,\ldots, m-1\}$ we have $1^j =
k^{j+1}$ for some $k \in \{2,\ldots,n\}$. In other words, it is a
chain of $m$ cycles. First, we have the following result.

\begin{restatable}{theorem}{chainnoweight}
\label{thm:necklace-noweight-nobonus}
Every coordination game on an unweighted \ochain{} has an improvement
path of length $\mathcal{O}(nm^2)$.
\end{restatable}

\begin{proof}[Proof sketch]
We provide a proof sketch, the details can be found in the
appendix. 

The idea behind the proof is to view the \ochain{} as a sequence of
simple cycles with bonuses. Here at most two nodes in each cycle have
non-trivial bonuses. We then apply Theorem~\ref{thm:TARK15} to
construct a finite improvement path for each cycle and argue that
these paths can be composed in a certain manner to obtain a finite
improvement path in the \ochain{} that terminates in a Nash
equilibrium.

Let $\{\mathcal{C}_j \mid j \in \{1,2, \ldots, m\}\}$ be the set of
simple cycles which constitute the \ochain{} $\mathcal{N}$. The
maximum in-degree of any node in $\mathcal{N}$ is two and in each
$\mathcal{C}_j$, there are at most two nodes $u$ and $v$ with
in-degree two with one of the incoming edges $x \to u$ from a node $x$
in $\mathcal{C}_{j+1}$ if $j < m$ and the other $y \to v$ from a node
$y$ in $\mathcal{C}_{j-1}$ if $j >1$. Given a joint strategy $s$, we
can view these external incoming edges into $\mathcal{C}_{j}$ as
bonuses to the nodes $u$ and $v$. That is, $\beta_j^s(u,c)=1$ if
$s_x=c$ and $0$ otherwise, $\beta_j^s(v,c)=1$ if $s_y=c$ and $0$
otherwise. For all $i \in \{1^j,\ldots,n^j\} \setminus \{u,v\}$, for
all $c$, $\beta_j^s(i,c)=0$.

For each $j \in \{1,2, \ldots, m\}$ and a joint strategy $s$, consider
the cycle $\mathcal{C}_j$ along with the bonus $\beta_j^s$. This
induces a coordination game on a (unweighted) simple cycle with
bonuses. By Theorem~\ref{thm:TARK15} such a coordination game is
weakly acyclic.

Given a joint strategy $s$, for $j \in \{1,\ldots,m\}$, call the node
$1^j$ a {\it break point} in $s$ if the following two conditions are
satisfied:
\begin{itemize}
\item[(C1)] $\forall k \leq j$ and $i \in \{1^k,\ldots                            
  n^k\}$, $s_{i}$ is a best response to $s_{-i}$,
\item[(C2)] $s_{1^j} = s_{n^j}$.
\end{itemize}

For a joint strategy $s$, let $\guard(s)$ be the largest $j
\in \{1, \ldots, m-1\}$ such that $1^j$ is a break point in $s$, if no
such $j$ exists then $\guard(s)=0$. 
Let $s^0$ be an arbitrary joint strategy in the game whose underlying
graph is $\mathcal{N}$. We construct a finite improvement path
inductively as follows. Initially, the improvement path consists of
the joint strategy $s^0$. Suppose we have constructed an improvement
path $\rho'$ such that $\lasts(\rho')=s'$. Choose the least $j \in
\{1,\ldots,m\}$ such that there is a node in $\mathcal{C}_j$ which is
not playing its best response in $s'$. Apply Theorem~\ref{thm:TARK15}
to the game induced by $\mathcal{C}_j$ and $\beta_j^{s'}$ to extend
the improvement path. 

We can then argue that each time a cycle $\mathcal{C}_{j+1}$ is
updated, either the number of cycles playing the best response
strictly goes up or if that quantity decreases, then the value of
$\guard$ strictly increases. Note that the value of $\guard$ is always
weakly increasing. Thus if we consider the pair, the value of $\guard$
and the number of cycles playing the best response, then this pair
under lexicographic ordering forms a progress measure for the specific
scheduling of nodes defined above. The value of $\guard$ is bounded by
$m-1$ and the number of cycles is bounded by $m$.

The improvement path constructed in Theorem~\ref{thm:TARK15} is of
length $\mathcal{O}(n)$. Each time the number of cycles playing the
best response increase, a single colour can be propagated down the
entire chain of cycles. In the worst case, the value of the guard can
increase by 1 at the end of each phase. Thus in the worst case, the
length of the improvement path that is constructed is
$\mathcal{O}(nm^2)$.
\end{proof}

\noindent{\bf Weighted \ochain{}.} We say that an \ochain{} is
weighted if at least one of the component cycle has an edge with
non-trivial weights (i.e.\ an edge with weight at least 2). We now
show that Theorem~\ref{thm:necklace-noweight-nobonus} can be extended
to the setting of weighted \ochain{}.

As in the proof of Theorem~\ref{thm:necklace-noweight-nobonus}, the
idea behind the proof is to view the weighted \ochain{} as a sequence
of {\it weighted} simple cycles with bonuses. The crucial observation
is that at most two nodes in each cycle have bonuses. We can then
apply Theorem~\ref{thm:cycle-2bonuses} to construct a finite
improvement path for each cycle and argue that these paths can be
composed in a specific manner.

Let $\{\mathcal{C}_j \mid j \in \{1,2, \ldots, m\}\}$ be the set of
simple cycles which constitute the \ochain{} $\mathcal{N}$. By the
definition of $\mathcal{N}$, for all $j \in \{1,\ldots,m-1\}$, the
node $1^j$ has indegree two with an edge $n^j \to 1^j$ of weight $w_1^j$
and an edge $(k-1)^{j+1} \to 1^j$ with weight $w_2^j$. To simplify the
presentation of the proof, we assume that for all $j$, $w_1^j \neq
w_2^j$. %
We first show the following two restricted results.

\begin{restatable}{lem}{wchainup}
\label{lm:weighted-chain1}
In the \ochain{} $\mathcal{N}$ consisting of cycles $\mathcal{C}_1,
\ldots, \mathcal{C}_m$, if for all $j \in \{1,2, \ldots, m-1\}$ we
have $w_1^j > w_2^j$ then $\mathcal{N}$ is weakly acyclic.
\end{restatable}

\begin{proof}[Proof sketch]
We provide a proof sketch, the details can be found in the appendix.
As in the proof of Theorem~\ref{thm:necklace-noweight-nobonus}, given
a joint strategy $s$, we can view the external incoming edges into
$\mathcal{C}_{j}$ as bonuses to the corresponding nodes. The only
difference in this case is that the value of the bonus instead of
being 1, is the weight of the corresponding edge. Let $\beta^s_j$
denote this bonus function. The cycle $\mathcal{C}_{j}$ along with
$\beta^s_j$ defines a coordination game on a weighted simple cycle
with at most two nodes having non-trivial bonuses. By
Theorem~\ref{thm:cycle-2bonuses}, such a coordination game is weakly
acyclic.

Let $s^0$ be an arbitrary joint strategy in the game whose underlying
graph is $\mathcal{N}$. We construct a finite improvement path
inductively as follows. Initially, the improvement path consists of
the joint strategy $s^0$. Suppose we have constructed an improvement
path $\rho'$ such that $\lasts(\rho')=s'$. Choose the least $j \in
\{1,\ldots,m\}$ such that there is a node in $\mathcal{C}_j$ which is
not playing its best response in $s'$. Apply
Theorem~\ref{thm:cycle-2bonuses} to the game induced by
$\mathcal{C}_j$ and $\beta_j^{s'}$ to extend the improvement
path. Since $w_1^j > w_2^j$ for all $j$, we can show that the partial
improvement path $\rho$ that is constructed in this manner satisfies
the following invariant:

\begin{itemize}
\item[(I)]Let $\lasts(\rho)=s$ and let $j$ be the largest index $j \in
  \{1, \ldots m-1\}$ such that for all $k \leq j$, $i \in \{1^k,\ldots
  n^k\}$, $s_{i}$ is a best response to $s_{-i}$. If $s_{1^k} \neq
  s_{n^k}$ then $s_{n^k} \not\in C(1^k)$.
\end{itemize}

The invariant asserts that if in the strategy $s$, the choice of the
nodes $1^k$ and its unique predecessor $n^k$ in $\mathcal{C}_k$ are
not the same then the colour chosen by $n^k$ is not in the available
colours for $1^k$. Using this, We can argue that the above procedure
terminates in a Nash equilibrium.
\end{proof}

\begin{restatable}{lem}{wchaindown}
\label{lm:weighted-chain2}
In the \ochain{} $\mathcal{N}$ consisting of cycles $\mathcal{C}_1,
\ldots, \mathcal{C}_m$, if for all $j \in \{1,2, \ldots, m-1\}$ we
have $w_1^j < w_2^j$ then $\mathcal{N}$ is weakly acyclic.
\end{restatable}

\begin{proof}
Let $s^0$ be an arbitrary joint strategy in the game whose underlying
graph is $\mathcal{N}$. We construct a finite improvement path
inductively as follows. Initially, the improvement path consists of
the joint strategy $s^0$. Suppose we have constructed an improvement
path $\rho'$ such that $\lasts(\rho')=s'$. Choose the greatest $j \in
\{1,\ldots,m\}$ such that there is a node in $\mathcal{C}_j$ which is
not playing its best response in $s'$. Apply
Theorem~\ref{thm:cycle-2bonuses} to the game induced by
$\mathcal{C}_j$ and $\beta_j^{s'}$ to extend the improvement
path. Since $w_1^j < w_2^j$ for all $j$, we can show that the partial
improvement path $\rho$ that is constructed in this manner satisfies
the following invariant:

\begin{itemize}
\item Let $\lasts(\rho)=s$ and let $j$ be the smallest index $j \in
  \{1, \ldots m-1\}$ such that for all $l \geq j$, $i \in \{1^l,\ldots
  n^l\}$, $s_{i}$ is a best response to $s_{-i}$. If $s_{1^l} \neq
  s_{(k-1)^{l+1}}$ then $s_{(k-1)^{l+1}} \not\in C(1^k)$.
\end{itemize}

Due to the invariant above and the fact that $w_1^j < w_2^j$ for all
$j$, we can use an argument very similar to that of the proof of
Lemma~\ref{lm:weighted-chain1}, to show that a finite improvement path
can be constructed.
\end{proof}

\begin{restatable}{theorem}{ochainweight}
\label{thm:necklace-nobonus}
Every coordination game on a weighted \ochain{} has an improvement
path of length $\mathcal{O}(nm^3)$.
\end{restatable}

\begin{proof}[Proof sketch]
Let $\mathcal{N}$ be the \ochain{} consisting of the sequence of
weighted cycles $\mathcal{C}_1, \ldots, \mathcal{C}_m$. The idea is to
split this sequence of cycles into blocks. A block $B_j$ is simply a
sequence of simple cycles in $\mathcal{N}$, say $\mathcal{C}_p,
\ldots, \mathcal{C}_l$ such that one of the following conditions hold,
\begin{itemize}
\item for all $k \in \{p,\ldots,l-1\}$ either $w_1^k > w_2^k$,
\item for all $k \in \{p,\ldots,l-1\}$, $w_2^k > w_1^k$.
\end{itemize}

We can then repeatedly apply Lemma~\ref{lm:weighted-chain1} and
\ref{lm:weighted-chain2} and compose the resulting improvement paths
in a specific manner to construct a finite improvement path for
$\mathcal{N}$.

The improvement path constructed by applying
Lemma~\ref{lm:weighted-chain1} and Lemma~\ref{lm:weighted-chain2} can
be of length $\mathcal{O}(nm^2)$. While composing this path we might
have to propagate colours down the chain. We can argue that
we always make progress by at least one block. Thus in worst case, the
length of the improvement path can be $\mathcal{O}(nm^3)$. 
The details can be found in the appendix.
\end{proof}

If we allow both weights and bonuses in the underlying graph
which constitutes an \ochain{}, then it follows from
Example~\ref{ex:noNE} that there are coordination games that do not
have a Nash equilibrium.

\medskip

\noindent{\bf Closed chain of cycles.} As earlier, let $\mathcal{C}_j$
be the cycle $1^j \to 2^j \ldots \to n^j \to 1^j$ for $j \in \{1,
\ldots,m\}$. Consider the structure in which for all $j \in
\{1,\ldots,m-1\}$, we have $1^j = k^{j+1}$ for some $k \in \{2,\ldots
n\}$ and $1^m=k^1$. In other words, instead of having a chain of
simple cycles, we now have a ``cycle'' of simple cycles. We can argue
that if these simple cycles are unweighted then the coordination game
whose underlying graph is such a structure remains weakly
acyclic. However, if we allow the simple cycles to have non-trivial
weights then the resulting game need not have a Nash equilibrium as
demonstrated in Example~\ref{ex:necklace-weight}. Note that to
construct a counter example (Figure~\ref{fig:necklace-weight}), we
only need three cycles each containing three nodes and a single edge
in each cycle with weight 2.

\begin{restatable}{theorem}{cchainnoweight}
\label{thm:necklace-cycle-noweights}
Every coordination game on an unweighted \cchain{} has an improvement
path of length $\mathcal{O}(nm^2)$.
\end{restatable}

\begin{proof}[Proof sketch]
Let $\{\mathcal{C}_j \mid j \in \{1,2, \ldots, m\}\}$ be the set of
simple cycles which constitute the graph $\mathcal{C}$. In each
$\mathcal{C}_j$, there is exactly two nodes with in-degree two. By the
definition of $\mathcal{C}$, the simple cycles $\mathcal{C}_1$ and
$\mathcal{C}_m$ share one node $1^m=k^1$ for some $k \in \{1, \ldots
n\}$. Let $s^0$ be an arbitrary joint strategy in the game whose
underlying graph is $\mathcal{C}$. The idea of the proof is to view
the sequence of cycles $\mathcal{C}_1, \ldots, \mathcal{C}_{m-1}$ as
an \ochain{}. By Theorem~\ref{thm:necklace-noweight-nobonus}, there is
a finite improvement path starting at $s^0$ and terminating in $s^1$
such that for all nodes in cycles $\mathcal{C}_1, \ldots,
\mathcal{C}_{m-1}$ are playing their best response in $s^1$. We can argue
that this path can be extended to a finite improvement path in
$\mathcal{C}$. 
The improvement path constructed by applying
Theorem~\ref{thm:necklace-noweight-nobonus} has length
$\mathcal{O}(nm^2)$. And this can be extended to a finite improvement
path in $\mathcal{C}$ with a constant number of updates of nodes in
the cycles. The details can be found in the appendix.
\end{proof}

\begin{example}
\label{ex:necklace-weight}
Consider the coordination game with the underlying graph given in
Figure~\ref{fig:necklace-weight}. Here, the nodes 4, 5, and 6 do not
have a choice of colours and so in any joint strategy they need to
choose the unique colour in their respective colour set. The set of
joint strategies that we need to consider is then the same as given in
Example~\ref{ex:noNE}. It follows that the game does not have a Nash
equilibrium.
\qed
\end{example}

\begin{figure}%
\centering
\tikzstyle{agent}=[circle,draw=black!80,thick, minimum size=1.8em,scale=0.8]
\begin{tikzpicture}[auto,>=latex',shorten >=1pt,on grid]
\R=1.3cm
\newcommand{\llab}[1]{{\small $\{#1\}$}}
\draw (90: \R) node[agent,label=right:{\llab{a,b}}] (1) {1};
\draw (90-120: \R) node[agent,label=right:{\llab{a,c}}] (2) {2};
\draw (90-240: \R) node[agent,label=left:{\llab{b,c}}] (3) {3};
\draw (30: \R) node[agent,label=right:\llab{a}] (4) {4};
\draw (30-120: \R) node[agent,label=right:{\llab{c}}] (5) {5};
\draw (30-240: \R) node[agent,label=left:{\llab{b}}] (6) {6};

    \draw[->] (1) -- (2) node [midway,left]{\small{$2$}};
    \draw[->] (3) -- (1) node [midway,right]{\small{$2$}};
    \draw[->] (2) -- (3) node [midway,above]{\small{$2$}};

\foreach \x/\y in {4/1,2/4,5/2,3/5,6/3,1/6} {
    \draw[->] (\x) to (\y);
}
\end{tikzpicture}
\caption{A weighted \cchain{} with no Nash equilibrium. \label{fig:necklace-weight}}
\end{figure}
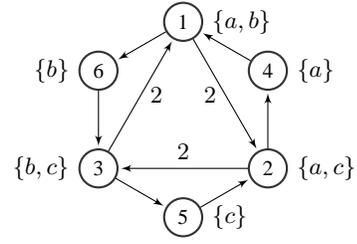

As in the case of simple cycles, we can show that unweighted \cchains
and \ochains are c-weakly acyclic. This implies the existence of
strong equilibria in coordination games played on such graph
structures.

\begin{restatable}{theorem}{secchain}
\label{thm:se-open-chain}
Every coordination game on an unweighted \cchain{}
has a c-improvement path of length $\mathcal{O}(nm^3)$.
\end{restatable}

\begin{proof}[Proof sketch]
We provide a proof sketch, the details can be found in the appendix.
Let $s$ be a \NE that this game reaches via an improvement path of
length $\mathcal{O}(nm^2)$ as constructed in Theorem
\ref{thm:necklace-cycle-noweights}.  If $s$ is not a strong
equilibrium, then there exists a coalition $K$ with a profitable
deviation, $s'$, from $s$.  Due to Lemma \ref{lem:unicolor-cycle}, the
coalition $K$ has to include at least one simple cycle, $\mathcal{C}$,
switching to the same colour in $s'$.  This can be one of the cycles
$\mathcal{C}_i$ or one of the two cycles going around the whole game
graph containing the set of nodes $A = \{1^j\ |\ j \in \{1,\ldots,m\} \}$. 
Note that $s'$ may not be a \NE but because the game is weakly
acyclic there is a finite improvement path which leads to a \NE $s''$
from $s'$.  We can argue that $s'_i = s''_i$ for all $i \in \mathcal{C}$.

Now again, if $s''$ is not a strong equilibrium, then there exists a
new coalition $K'$ with a profitable deviation, $s'''$, from $s''$.
We can show that either $\mathcal{C} = A$ or no node from
$\mathcal{C}$ can be part of $K'$. This implies that we can construct
a c-improvement path by appropriately composing the improvement paths
from Theorem \ref{thm:necklace-cycle-noweights} along with deviations
by simple cycles.  At least one simple cycle changes colour in each
such deviation and none of its nodes change colour afterwards.  This
shows that the number of non-unilateral coalition deviation is at most
equal to the number of different simple cycles in the game graph,
which is equal to $m$.  Thus there is a c-improvement path of length
$\mathcal{O}(nm^3)$.
\end{proof}

\begin{restatable}{cor}{seochain}
\label{cor:se-open-chain}
Every coordination game on an unweighted \ochain{}
has a c-improvement path of length $\mathcal{O}(nm^3)$.
\end{restatable}

\begin{proof}
Any \ochain{} can be converted to a \cchain{} by adding one additional colour, $c^*$, and one
simple cycle with four nodes: one node from $\mathcal{C}_1$ different from $1^1$ and $n^1$,
one node from $\mathcal{C}_m$ different from $1^m$ and $n^m$,
two extra nodes between these two with $c^*$ as the only colour available to them.
It is easy to see that a coordination game played on this 
\cchain{}
has essentially the same 
behaviour as a game played on the original \ochain{}.
In particular there is a one-to-one mapping between their c-improvement paths.
\end{proof}

Finally, so far we assumed that we know the decomposition of the game
graph into a chain of cycle in advance.  In general the input may be
an arbitrary graph and we would need to find this decomposition first.
Fortunately this can be done in linear time.

\begin{restatable}{prop}{checkchain}
\label{pr:detect-chains}
Checking whether a given graph $G$ is a \ochain{} or \cchain{}, and if so
partitioning $G$ into simple cycles $\mathcal{C}_1, \ldots, \mathcal{C}_m$ can be done in $\calO(|G|)$.
\end{restatable}

\begin{proof}
Assuming that $G$ is a \cchain{}, the set $A = \{1^j | 1\leq j \leq
m\}$ is just the set of all nodes in $G$ with outdegree $2$. %
denoted by $B$.  We perform a depth first search starting from any
node of $G$ and list the nodes with outdegree $2$ as we encounter
them.  We identify the $j$-th node on this list with $1^j$.
We can then easily identify the remaining nodes in each cycle $C_j$
for $1\leq j \leq m$. Argument is similar for \ochains.
\end{proof}

%% file: partition.tex
\section{Simple cycles with cross-edges}
\label{sec:partion-cycles}
In this section we consider coordination games whose underlying graph
forms simple cycles along with some additional ``non-cyclic'' edges
between nodes. We say that the graph $G=(V,E)$ is a simple cycle with
cross-edges if $V=\{1,2,\ldots,n\}$ and the edge set $E$ can be
partitioned into two sets $E_c$ and $E_p$ such that $E_c=\{i \to
i\oplus 1 \mid i \in \{1,\ldots n\}\}$ and $E_p=E \setminus E_c$. In
other words, $E_c$ contains all the cyclic edges and $E_p$ all the
additional cross-edges in $G$.

The results in the previous section show that simple cycles are quite
robust in terms of maintaining the property of being weakly
acyclic. Even with weighted edges and chains of simple cycles, the
resulting coordination games remain weakly acyclic. In this section,
we study the same question: whether simple cycles with cross-edges are
weakly acyclic. We first show that if we allow arbitrary (unweighted)
cross-edges, then there are games that may not have a Nash equilibrium
(Example~\ref{ex:part-cycle-noNE}). We then identify a restricted
class of cycles with cross-edges for which the game is weakly acyclic.

\begin{example}
\label{ex:part-cycle-noNE}
Consider the graph $G'$ which we obtain by adding the following edges
to the graph in Figure~\ref{fig:example-graph}: $6 \to 7$, $4 \to 8$
and $5 \to 9$. Thus $G'$ defines a simple cycle: $1 \to 4 \to 8 \to 2
\to 5 \to 9 \to 3 \to 6 \to 7 \to 1$ along with the cross-edges
represented in Figure~\ref{fig:example-graph} (the nodes in $G'$ can
be easily renamed if required to form the cyclic ordering $1 \to 2
\ldots 8 \to 9$). Note that in the resulting graph $G'$, for any joint
strategy, the payoff for node 7 is always 0 since $C(7)$ and $C(6)$
are disjoint. Same holds for node 8 and node 9. Also, note that the
best response for nodes 4, 5 and 6 is to always select the same
strategy as nodes 1, 2 and 3 respectively. Therefore, to show that the
game does not have a Nash equilibrium, it suffices to consider the
strategies of nodes 1, 2 and 3. We can denote this by the triple
$(s_1, s_2, s_3)$. The joint strategies are then the same as those
listed in Example~\ref{ex:noNE}. It follows that the game does not
have a Nash equilibrium.
\qed
\end{example}

\noindent {\bf Partition-cycle.} Let $G=(V,E)$ be a simple cycle with
cross-edges where $E=E_c \cup E_p$. 
We call $G$ a \emph{partition-cycle} if $(V, E_c)$ forms a simple
cycle and the vertex set $V$ can be partitioned into two sets $V_T$
and $V_B$ such that $V_T,V_B \neq \emptyset$ and the following
conditions are satisfied: $E_p \subseteq V_T \times V_B$, 
\begin{itemize}
\item $E_c \cap (V_T \times V_T)$ forms a path in $(V, E_c)$,
\item $E_c \cap (V_B \times V_B)$ forms a path in  $(V, E_c)$.
\end{itemize}

\begin{figure}%
\centering
\tikzstyle{agent}=[circle,draw=black!80,thick, minimum size=2em,scale=.8]
\begin{tikzpicture}[auto,>=latex',shorten >=1pt,on grid]
\newcommand{\llab}[1]{{\small $\{#1\}$}}
\node[agent,label=above:{\llab{b,c}}] (1) {1};
\node[agent,right = 1.7 of 1, label=above:{\llab{b,c}}] (2) {2};
\node[agent,right = 1.7 of 2, label=above:{\llab{b}}] (3) {3};
\node[agent,right = 1.7 of 3, label=above:{\llab{c}}] (4) {4};
\node[agent, below right = .6 and 1.7 of 4, label=right:{\llab{a}}] (5) {5};
\node[agent,below = 1.2 of 4, label=below:{\llab{a,b}}] (6) {6};
\node[agent,left = 1.7 of 6, label=below:{\llab{a,c}}] (7) {7};
\node[agent,left = 1.7 of 7, label=below:{\llab{b,c}}] (8) {8};
\foreach \x/\y in {2/3, 3/4, 6/7, 7/8}{
    \draw[->] (\x) to (\y);
}

\foreach \x/\y in {4/5, 5/6, 8/1}{
\draw[->, bend left, bend angle=45] (\x) to (\y);
}

\foreach \x/\y in {1/2, 3/8, 4/7, 2/6, 1/6}{
    \draw[->] (\x) to (\y);
}

\end{tikzpicture}
\caption{A partition-cycle. \label{fig:partcycle}}
\end{figure}
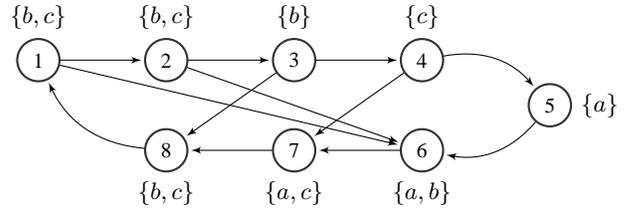

\begin{example}
The directed graph in Figure~\ref{fig:partcycle} is an example of a
partition-cycle. One possible partition of the vertex set would be
$V_T =\{1,2,3,4,5\}$ and $V_B=\{6,7,8\}$. The edge set $E_c$
consists of the edges $1 \to 2, 2 \to 3, \ldots, 8 \to 1$
whereas $E_p=\{1 \to 6, 2 \to 6, 3 \to 8, 4
\to 7\}$.  
\qed
\end{example}

We first show that every coordination game whose underlying graph is
an unweighted partition cycle is weakly acyclic. For the sake of
simplicity, we fix the following notation: the partition-cycle is
given by $G=(V,E)$ where $V=\{1,\ldots n\}$, $V_T=\{1,2, \ldots, k\}$
and $V_B=\{k+1,k+2, \ldots, n\}$.
If $E_p =
\emptyset$ then we get a simple cycle without cross-edges on $n$
nodes. For $i \in V_B$, $c \in C(i)$ and a joint strategy $s$, let
$\mathcal{S}(i,c,s)=\{j \in V_T \mid j \to i \text{ and } s_j=c\}$. We
also define the set $\mathit{MC}(i,s)=\{c \in C(i) \mid
|\mathcal{S}(i,c,s)| \geq |\mathcal{S}(i,c',s)| \text{ for all } c'
\in C(i)\}$. Given a player $i$ and a joint strategy of the other players
$s_{-i}$ let $\mathit{BR}(i, s_{-i})$ denote the set of best responses
of player $i$ to $s_{-i}$.

\begin{restatable}{theorem}{pcyclenoweightnobonus}
\label{thm:partcycle-noweight-nobonus}
Every coordination game without bonuses on an unweighted
partition-cycle has an improvement path of length $\mathcal{O}(n(n-k))$.
\end{restatable}

\begin{proof}
Consider an initial joint strategy $s^0$. We construct a finite
improvement path starting in $s^0$ as follows. We proceed around the
cycle and consider the players $1, 2, \ldots ,n$ in that order. For
each player $i$, in turn, for the corresponding joint strategy $s$, if
$s_i$ is not a best response to $s_{-i}$, we update it to a best
response respecting the following property:

\begin{itemize}
\item[(P1)] If $s_{i \ominus 1} \in \mathit{BR}(i, s_{-i})$ and there
  exists a $c \in \mathit{MC}(i,s)$ such that $p_i(c,s_{-i})=p_i(s_{i
    \ominus 1}, s_{-i})$ then player $i$ switches to $c$ (in
  this case $c \in \mathit{BR}(i, s_{-i})$ as well).
\end{itemize}

Let $s^1$ be the resulting joint strategy at the end of the first
round. It follows that the players $2,\ldots,n$ are playing their best
response in $s^1$. If $s^1$ is a Nash equilibrium then the improvement
path is constructed. If not then the only player who is not playing
its best response is player $1$. This implies that $s^1_n \neq
s^0_n$. Let $l_1$ be the least index in $V_B=\{k+1,\ldots n\}$ such
that for all $j \in \{l_1, \ldots, n\}$, $s^1_{l_1} \neq s^0_{l_1}$
and $s^1_{l_1}=s^1_n$. Let $X=\{l_1, l_1+1, \ldots, n\}$. Note that $X
\neq \emptyset$ since $n \in X$. We repeatedly let players update to
their best response strategies in the cyclic order in multiple
rounds. We can argue that in each round $|X|$ strictly increases. By
definition, $|X| \leq |V_B|$ and therefore the improvement path
constructed in this manner eventually terminates in a Nash
equilibrium.  

In the second round starting at the joint strategy $s^1$, we let
players update to their best response following the cyclic order $1,2,
\ldots k$. Let $s^2$ be the resulting joint strategy. Note that in
this sequence, if a player is not playing its best response then the
best response strategy is simply to switch to the current strategy of
its unique predecessor on the cycle (recall that all nodes in $V_T$
have exactly one incoming edge). Thus the only colour which is
propagated is $s^1_n$. Now starting at $s^2$, let players update to
their best response following the cyclic order $k+1, \ldots n$ and let
$s^3$ be the resulting joint strategy. If $s^3$ is a Nash equilibrium
then we have a finite improvement path. If not, then player 1 is the
unique player not playing its best response and $s^3_n \neq
s^1_n$. 
We know that for all $j \in X$, $s^1_j=s^1_n$. By the above argument
we also have $|\mathcal{S}(j, s^1_n, s^3)| \geq |\mathcal{S}(j, s^1_n,
s^1)|$. Thus if $s^3_n \neq s^1_n$ then for all $j \in X$, $s^3_j =
s^3_n$. Now consider the node $l_1$ and let $t$ and $t'$ be the joint
strategies in the improvement path constructed where
$t=(s^1_{l_1},t_{-l})$ and $t'=(s^3_{l_1},t_{-l})$. For all $m \in
V_T$, we have $t_m=s^2_m(=s^3_m)$ and $|\mathcal{S}(l_1, s^1_{l_1},
s^3)| \geq |\mathcal{S}(l_1, s^1_{l_1}, s^1)|$. Thus if $t_{l_1}$ is
not a best response of player $l_1$ then $s^3_{l_1}=s^3_{{l_1}-1}$ and
$s^3_{{l_1}-1} \neq s^1_{{l_1}-1}$. Now let $l_2$ be the least index
in $V_B=\{k+1,\ldots n\}$ such that for all $j \in \{l_2, \ldots,
n\}$, $s^3_{l_1} \neq s^1_{l_1}$ and $s^3_{l_1}=s^3_n$. Let $X'=\{l_2,
l_2+1, \ldots, n\}$. Clearly, $l_2 < l_1$ and therefore, $|X'| > |X|$
and $X \subseteq X'$. Let $X:=X'$ and we repeat this process. In each
successive round, $|X|$ strictly increases and by definition, $|X|
\leq |V_B|$. Therefore, in at most $|V_B|$ rounds, either we reach a
Nash equilibrium or we reach a joint strategy $s'$ where for all $j,m
\in V_B$, $s'_j=s'_m$. In this case we go around in the cyclic order
$1,2, \ldots,k$ and update players to their best response. As earlier
we can argue that the only colour which is propagated is $s'_n$ and
therefore this improvement path terminates in a Nash equilibrium.

The above proof shows that starting from the second round, the size of
the set $X$ strictly increases and we know that $|X| \leq |V_B| = n-k$.
Each time, in the worst case, we might have to update all the
nodes in the cyclic order. Thus in the worst case the length of this
improvement path is at most $\mathcal{O}((n-k)\cdot n)$
\end{proof}

From Theorem~\ref{thm:cycle-2bonuses}, we know that simple cycles even
with weighted edges are weakly acyclic.
However, partition-cycles with weighted edges need not always have a Nash equilibrium (Example~\ref{ex:partcycle-noNE}). 
On the other hand, we show in
Theorem~\ref{thm:partcycle-noweight-bonus} that unweighted
partition-cycles with bonuses remain weakly acyclic. Thus
Theorem~\ref{thm:TARK15} can be extended to partition-cycles. 

\begin{example}
\label{ex:partcycle-noNE}
Consider the partition-cycle $G$ given in Figure~\ref{fig:partcycle}
and suppose we add weight 2 to edges $6 \to 7$ and $7 \to 8$. The
resulting game does not have a Nash equilibrium.  Note that in any
joint strategy, nodes 3, 4 and 5 have to choose the colour $b$, $c$
and $a$ respectively. Therefore, it suffices to consider strategies of
nodes 6, 7, 8, 1 and 2. Also note that in any joint strategy $s$, the
best response for players 1 and 2 is $s_8$ (the strategy of player 8
in $s$). Thus we can also restrict attention to joint strategies $s$
in which $s_1=s_2=s_8$. So let us denote a joint strategy $s$ by the
triple $(s_6,s_7,s_8)$. Below we list all such joint strategies and we
underline a strategy that is not a best response:
$(\underline{a},a,b)$, $(a,a,\underline{c})$, $(a,c,\underline{b})$,
$(a,\underline{c},c)$, $(b,\underline{a},b)$, $(\underline{b},a,c)$,
$(b,c,\underline{b})$ and $(\underline{b},c,c)$.
\qed
\end{example}

\begin{restatable}{theorem}{pcyclenoweight}
\label{thm:partcycle-noweight-bonus}
Every coordination game with bonuses on an unweighted partition-cycle
has an improvement path of length $\mathcal{O}(kn(n-k))$.
\end{restatable}
\begin{proof}[Proof sketch]
The main idea is to enforce the players to update their strategy based
on a specific priority over colours induced by the bonuses.  For each
node in $V_T$ to satisfy the priority of updates over colours induced
by the bonuses, we might have to cycle through each node and construct
the improvement path as given in the proof of
Theorem~\ref{thm:partcycle-noweight-nobonus}. Thus in the worst case,
the length of the improvement path which is constructed is
$\mathcal{O}(k\cdot(n-k)n)$. 
Details are provided in the appendix.
\end{proof}

Note that Example~\ref{ex:partcycle-noNE} shows that with just
two weighted edges between nodes in $V_B$, it is possible to construct
games which may not have a Nash equilibrium. We now show that if the
weights are only present on edges between nodes in $V_T$ or on the
cross-edges $E_p$ then the resulting game remains weakly acyclic. If
we allow bonuses on nodes then we can add weights to the cross-edges
$E_p$ and the resulting game remains weakly acyclic. On the other
hand, from Example~\ref{ex:noNE} we already know that if we allow both
weights and bonuses, even without cross-edges, there are graphs in
which the resulting game need not have a Nash equilibrium.

Given a partition cycle $G=(V_T \cup V_B, E_c \cup E_p)$, let $E_T =
(V_T \times V_T) \cap E_c$. That is, the set $E_T$ consists of all the
cyclic edges between nodes in $V_T$.

\begin{restatable}{theorem}{pcyclewtopcross}
\label{thm:partcycle-weightsTopCross-nobonus}
Every coordination game without bonuses on a partition-cycle with
weights on edges in $E_T \cup E_p$ is weakly acyclic.
\end{restatable}

\begin{proof}
Let $G=(V_T\cup V_B,E_c \cup E_p)$ be a partition-cycle and $E_T =
(V_T \times V_T) \cap E_c$. We first show that for each weighted edge
in $E_p$ we can add a set of unweighted edges and obtain a new
partition-cycle $G'$ such that every improvement path in $G'$ can be
converted into an improvement path in $G$. Let $u \to v$ be an edge in
$E_p$ with weight $w$. Note that by definition of $G$, $u \in V_T$ and
$u \in V_B$. Let $x \to u$ and $u \to y$ be the cyclic edges in $E_c$
associated with the node $u$. We replace the node $u$ with $w$ new
nodes $u_1, \ldots,u_w$ and for all $j \in \{1, \ldots, w\}$ we set
$C(u_j)=C(u)$. We also add the following unweighted edges to the edge
set $E$.  For all $j \in \{1, \ldots, w-1\}$, $u_j \to u_{j+1} \in
E_c$, $u_j \to v \in E_p$, $u_w \to v \in E_p$ and $\{x \to u_1, u_w
\to y\} \subseteq E_c$. In any joint strategy $s$, the best response
of nodes $u_2, \ldots u_w$ would be to choose the same colour as
$u_1$. Which implies the following: the node $v$ had an incoming edge
of weight $w$ supporting the colour $s_u$ in $G$ iff in the modified
graph in any joint strategy in which the nodes $u_2, \ldots, u_w$ are
playing their best response, the node $v$ has $w$ edges supporting the
colour $s_u$.

The proof of Theorem~\ref{thm:partcycle-noweight-nobonus} shows that
it is possible to construct a finite improvement path by updating
players in the cyclic order. A crucial property which was used 
is that in each successive rounds, while updating players in $V_T$,
the only colour which is propagated is $s_n$, the colour chosen by
node $n$ in the end of the previous round. Even if the edges in $E_T$
are weighted, the property continues to hold since the best response
for each node $i \in V_T$ is still to choose the same colour as it
unique predecessor $i \ominus 1$ on the cycle, provided the colour is
in $C(i)$. Note that the edges in $E_c \setminus E_T$ are
unweighted. Thus by using a similar argument as in the proof of
Theorem~\ref{thm:partcycle-noweight-nobonus}, we can conclude that the
game is weakly acyclic.
\end{proof}

\begin{theorem}
\label{thm:partcycle-weightsCross-bonus}
Every coordination game with bonuses on a partition-cycle with weights
on edges in $E_p$ is weakly acyclic.
\end{theorem}

\begin{proof}
Each weighted edge in $E_p$ can be converted into a set of unweighted
edges such that the resulting graph $G'$ is still a
partition-cycle. From every finite improvement path in the
coordination game whose underlying graph is $G'$, we can construct a
finite improvement path in $G$. Thus by
Theorem~\ref{thm:partcycle-noweight-bonus}, the result follows.
\end{proof}

Finally, we assumed that the decomposition and ordering of the nodes
in the input partition-cycle graph is given in advance. The
decomposition can be computed in linear time as well.
\begin{proposition}
Checking whether a given graph $G$ is a partition-cycle and if so finding its $V_T$, $V_B$ 
and suitable ordering on these subsets of nodes can be done in $\calO(|G|)$.
\end{proposition}
\begin{proof}
Note that the ordering of $G = (V,E)$ we are looking for defines a Hamiltonian path in $G$ 
with particular properties.
We start by selecting only the nodes in $G$ with outdegree $1$; 
these are all the nodes that can potentially be in $V_B$.
Next, we look at the graph $G' = (V_B, E \cap V_B  \times V_B)$.
First, we remove any edges from $G'$ that form a cycle using, e.g. 
depth-first search. We then topologically sort the resulting DAG.
We obtain several disjoint paths $B_1, \ldots, B_k$ as candidates for $V_B$.
We assume that there is at least one cross-edge in $G$, because otherwise the problem is trivial.
We check to which of these disjoint paths this cross-edge leads to and we set that path as $V_B$
and the rest of the nodes are set as $V_T$.
The order on $V_B$ is given by the topological order.
We then look at $G'' = (V_T, E \cap V_T  \times V_T)$.
If $G''$ has a cycle then $G$ is a not a partition-cycle.
Otherwise, topologically sorting $G''$ gives us the order of nodes in $V_T$. 
Finally, it is straightforward to test whether $V_T$ and $V_B$ satisfy the 
remaining requirements for the graph $G$ to be a partition-cycle.
\end{proof}

%% file: no-way-to-se.tex
\section{Conclusions}
\label{sec:conclusions}

\begin{figure}
\centering
\tikzstyle{agent}=[circle,draw=black,thick, minimum size=2em,scale=0.7]
\tikzstyle{small}=[scale=0.9]
\begin{tikzpicture}[auto,>=latex',shorten >=1pt,on grid,scale=0.86]
\newdimen\R
\R=1.7cm
\newcommand{\llab}[1]{{\small $\{#1\}$}}
\newcommand{\lla}{\llab{\underline{a},b,c}}
\newcommand{\llb}{\llab{a,\underline{b},c}}
\newcommand{\llc}{\llab{a,b,\underline{c}}}
\draw (90: \R) node[agent,label=right:{\lla}] (1) {1};
\draw (90-120: \R) node[agent,label={[label distance=-4pt]below left:{\lla}}] (2) {2};
\draw (90-240: \R) node[agent,label={[label distance=-4pt]below right:{\llb}}] (3) {3};
\draw (90+15: 2*\R) node[agent,label=left:{\llb}] (4) {4};
\draw (90-15: 2*\R) node[agent,label=right:{\lla}] (5) {5};
\draw (-15: 2*\R) node[agent,label=right:{\lla}] (6) {6};
\draw (-15-30: 2*\R) node[agent,label=right:{\llc}] (7) {7};
\draw (180+45: 2*\R) node[agent,label=left:{\llc}] (8) {8};
\draw (180+15: 2*\R) node[agent,label=left:{\llb}] (9) {9};
\draw (30: 1.42*\R) node[agent,label=right:{\lla}] (A) {A};
\draw (120+30: 1.42*\R) node[agent,label=left:{\llb}] (B) {B};
\draw (240+30: 1.42*\R) node[agent,label={[label distance=-2pt]above:{\llc}}] (C) {C};
\foreach \x/\y/\w in {1/2/2,3/1/2} { %
    \draw[->, thick] (\x) to node[small] {\w} (\y) ;
}
\draw[->,thick] (2) to node[small,above] {2} (3);

\foreach \x/\y/\w in {4/1/2,1/5/3,2/6/2,2/7/3,8/3/2,9/3/3} { %
    \draw[<->, thick] (\x) to node[small] {\w} (\y) ;
}
\foreach \x/\y/\w in {5/A/3,A/6/2,B/4/2,9/B/3,8/C/2,C/7/3} { %
    \draw[<->, thick] (\x) to node[small] {\w} (\y) ;
}
\end{tikzpicture}
\caption{\label{fig:no-way-to-se}
A coordination game with trivial strong equilibria unreachable from the given initial joint strategy. 
}
\vskip-0.8em
\end{figure}

We presented natural classes of graphs for which coordination games
have improvement or c-{im\-prove\-ment} paths of polynomial size.  We
also showed that for most natural extensions of these classes, the
resulting coordination game may not even have a Nash equilibrium.
Note that although we defined bonuses as natural numbers, our results
also hold for any integer bonuses, because after increasing all
bonuses by a fixed amount, all players' incentives stay the same. 

In general, local search may not be an efficient technique to find a
Nash equilibrium or a strong equilibrium in coordination games
even when the game graph is strongly connected. 
In fact, a coordination game can have trivial strong equilibria which
cannot be reached from some of its initial joint strategies. For
example, the game in Figure \ref{fig:no-way-to-se} has three trivial
strong equilibria in which all players pick the same colour. However,
every improvement or c-improvement path from the initial joint
strategy (given by the underlined strategies) is infinite.  Moreover,
although the game graph is weighted, the weighted edges
can easily be replaced by unweighted ones
just by adding auxiliary nodes
(see Example \ref{ex:no-way-to-se} in the appendix).
Therefore, the non-existence of a finite improvement or c-improvement
path in coordination games even for strongly connected unweighted graphs does not imply
the non-existence of \NEs or strong equilibria.

In proving our results, we used various generalised potential
techniques, and exploited structural properties of the classes of
graphs studied.  It would be interesting to see whether there is a
common progress measure that works for all the classes of
graphs that we consider as well as for more general ones.  In
particular, we conjecture that coordination games on unweighted graphs
with indegree at most two are c-weakly acyclic. Extensive computer
simulations seem to support this conjecture.  This class of graphs
strictly generalises the unweighted \ochains and \cchains that we
showed to be c-weakly acyclic.
We also leave open the existence of finite c-improvement paths in weighted open chains of cycles and partition-cycles. 
Although they seem likely to exist, 
unicoloured simple cycles introduced by coalition deviations from \NEs
can disappear when trying to reach a new \NE after them,
so a detailed analysis of the interplay between these two steps is required to prove 
their c-weak acyclicity. 

\subsection*{Acknowledgements}
We are grateful to Krzysztof Apt for useful discussions. Sunil Simon
was supported by the Liverpool-India fellowship provided by the
University of Liverpool. Dominik Wojtczak was supported by EPSRC grant EP/M027651/1.

%% file: appendix.tex
\section*{Appendix A -- Sequence of simple cycles}
\chainnoweight*

\begin{proof}
Let $\{\mathcal{C}_j \mid j \in \{1,2, \ldots, m\}\}$ be the set of
simple cycles which constitute the \ochain{} $\mathcal{N}$. The
maximum in-degree of any node in $\mathcal{N}$ is two and in each
$\mathcal{C}_j$, there are at most two nodes $u$ and $v$ with
in-degree two with one of the incoming edges $x \to u$ from a node $x$
in $\mathcal{C}_{j+1}$ if $j < m$ and the other $y \to v$ from a node
$y$ in $\mathcal{C}_{j-1}$ if $j >1$. Given a joint strategy $s$, we
can view these external incoming edges into $\mathcal{C}_{j}$ as
bonuses to the nodes $u$ and $v$. That is, $\beta_j^s(u,c)=1$ if
$s_x=c$ and $0$ otherwise, $\beta_j^s(v,c)=1$ if $s_y=c$ and $0$
otherwise. For all $i \in \{1^j,\ldots,n^j\} \setminus \{u,v\}$, for
all $c$, $\beta_j^s(i,c)=0$.

For each $j \in \{1,2, \ldots, m\}$ and a joint strategy $s$, consider
the cycle $\mathcal{C}_j$ along with the bonus $\beta_j^s$. This
induces a coordination game on a (unweighted) simple cycle with
bonuses. By Theorem~\ref{thm:TARK15} such a coordination game is
weakly acyclic.

Given a joint strategy $s$, for $j \in \{1,\ldots,m\}$, call the node
$1^j$ a {\it break point} in $s$ if the following two conditions are
satisfied:
\begin{itemize}
\item[(C1)] for all $k \leq j$, for all $i \in \{1^k,\ldots                            
  n^k\}$, $s_{i}$ is a best response to $s_{-i}$,
\item[(C2)] $s_{1^j} = s_{n^j}$.
\end{itemize}

For a joint strategy $s$, let $\guard(s)$ be the largest $j
\in \{1, \ldots, m-1\}$ such that $1^j$ is a break point in $s$, if no
such $j$ exists then $\guard(s)=0$. 

Let $s^0$ be an arbitrary joint strategy in the game whose underlying
graph is $\mathcal{N}$. We construct a finite improvement path
inductively as follows. Initially, the improvement path consists of
the joint strategy $s^0$. Suppose we have constructed an improvement
path $\rho'$ such that $\lasts(\rho')=s'$. Choose the least $j \in
\{1,\ldots,m\}$ such that there is a node in $\mathcal{C}_j$ which is
not playing its best response in $s'$. Apply Theorem~\ref{thm:TARK15}
to the game induced by $\mathcal{C}_j$ and $\beta_j^{s'}$ to extend
the improvement path. 

In other words, the procedure works as follows: Suppose there is a
node in $\mathcal{C}_1$ which is not playing its best response in
$s^0$. Start with the coordination game induced by $\mathcal{C}_1$ and
$\beta_1^s$. By Theorem~\ref{thm:TARK15}, there is a finite
improvement path that terminates in a joint strategy $s^1$ such that
for all $i \in \{1^1,\ldots n^1\}$, $s^1_{i}$ is a best response to
$s^1_{-i}$. 

Now suppose we have constructed a partial improvement path $\rho$
where $\lasts(\rho)=s$ and let $j$ be the largest index $j \in \{1,
\ldots m-1\}$ such that for all $k \leq j$, $i \in \{1^k,\ldots
n^k\}$, $s_{i}$ is a best response to $s_{-i}$.
Consider the coordination game induced by the $\mathcal{C}_{j+1}$ and
$\beta_{j+1}^{s^j}$. By Theorem~\ref{thm:TARK15}, there is a finite
improvement path that terminates in a joint strategy $s^{j+1}$ such
that for all $i \in \{1^{j+1},\ldots n^{j+1}\}$, $s^{j+1}_{i}$ is a
best response to $s^{j+1}_{-i}$. Now since the two cycles
$\mathcal{C}_{j}$ and $\mathcal{C}_{j+1}$ share a node,
i.e.\ $1^j=k^{j+1}$. It is possible that in the joint strategy
$s^{j+1}$, the node $2^j$ is not playing its best response any
longer. To avoid multiple subscripts, let us denote the node $2^j$ by
$2$, $1^j$ by $1$. So we have
that $s^{j+1}_{2}$ is not a best response to $s^{j+1}_{-2}$. Note that
by assumption the node $2^j$ was playing its best response in the
joint strategy $s^j$. And the only node in $\mathcal{C}_j$ that could
possibly change its strategy in $s^{j+1}$ is $1^j$. 
Assume that node 2 has a unique predecessor (or the in-degree of node
2 is 1). Then we also have, $s^{j+1}_{1} \neq s^{j+1}_{2}$,
$s^{j+1}_{1} \in C(s^{j+1}_{2})$ and $p_{2}(s^{j+1}_{1}, s^{j+1}_{-2})
\geq p_{2}(c, s^{j+1}_{-2})$ for all $c \in C(2^j)$. Let the node
$2^j$ switch to the colour $s^{j+1}_{1}$. We then update the nodes in
the cyclic order in $\mathcal{C}_j$ successively if they are not
playing their best response. We then do the same procedure for cycles
in the order $\mathcal{C}_{j-1}, \ldots,
\mathcal{C}_{\guard(s^{j+1})}$.

Now suppose in this sequence of updates the only colour which is
propagated is $s^{j+1}_{1}$. Then we have reached a joint strategy in
which all nodes on cycles $\mathcal{C}_{j+1}, \mathcal{C}_{j}, \ldots,
\mathcal{C}_{\guard(s^{j+1})}$ are playing their best response. So the
number of cycles playing their best response has strictly
increased. If while propagating down the sequence $\mathcal{C}_{j-1},
\ldots, \mathcal{C}_{1}$ a new colour is introduced then note that
this new colour can only be introduced by a node of indegree
2. Suppose the first instance in the improvement path a new colour is
introduced is by node $1^l$ for $l < j$. Let $s^1$ be the joint
strategy before node $(k-1)^{l+1}$ updates to its best response and
$s^2$ be the joint strategy before node $1^l$ changes to the colour $c
\neq s^{j+1}_{1}$ (which is its best response). Recall that the node
$1^l$ is same as the node $k^{l+1}$ for some $k \in \{1,\ldots,n\}$.
Since we proceed in the cyclic order in $\mathcal{C}_{l}$, we know
that in $s^1$ node $1^l$ was playing a best response where as by
assumption in $s^2$ node $1^l$ is not. The only difference between
$s^1$ and $s^2$ is in the strategy of node $(k-1)^{l+1}$ and by
assumption, the strategy of node $(k-1)^{l+1}$ in $s^2$ is same as
$s^{j+1}_{1}$. Since $c \neq s^{j+1}_{1}$, $c \neq s^2_{1^l}$ and $c$
is a best response for node $1^l$ to the joint strategy $s^2_{-1^l}$,
it implies that $s^2_{n^l}=c$. Let $s^3$ be the joint strategy
obtained from $s^2$ by having $s^3_{1^l}=c$.  This implies that node
$1^l$ satisfies condition (C2) in the joint strategy $s^3$. If it also
satisfies condition (C1) then we have identified a break point. It
also follows that $\guard(s^3) > \guard(s^{j+1})$ and we have strictly
reduced the number of cycles we need to consider.

If condition (C1) is not satisfied in $s^3$ then the only node not on
its best response in $\mathcal{C}_l$ is $2^l$. Apply the same
propagation and let $1^k$ be the last node in this sequence which
introduces a new colour $c'$ and let $s^4$ be the joint strategy
obtained after node $1^k$ switches. By the same argument it holds that
condition (C2) is satisfied by node $1^k$ in the joint strategy
$s^4$. Now update the nodes in the cyclic order and in the sequence
$\mathcal{C}_k, \mathcal{C}_{k-1}, \ldots,
\mathcal{C}_{\guard(s^{j+1})}$. It can be verified that the only
colour propagated is $c'$. Let $s^5$ be the resulting strategy in
which all the nodes in the cycles $\mathcal{C}_k, \mathcal{C}_{k-1},
\ldots, \mathcal{C}_{\guard(s^{j+1})}$ are playing their best
response. This implies that the node $1^k$ is a break point in
$s^5$. Thus we have $\guard(s^5) > \guard(s^{j+1})$ and we can repeat
the same procedure inductively for the cycles
$\{\mathcal{C}_{\guard(s^5)+1}, \ldots, \mathcal{C}_m\}$.

The improvement path constructed in Theorem~\ref{thm:TARK15} is of
length $\mathcal{O}(n)$. Each time the number of cycles playing the
best response increase, a single colour can be propagated down the
entire chain of cycles. In the worst case, the value of the guard can
increase by 1 at the end of each phase. Thus in the worst case, the
length of the improvement path that is constructed is
$\mathcal{O}(nm^2)$.
\end{proof}

\wchainup*

\begin{proof}
As in the proof of Theorem~\ref{thm:necklace-noweight-nobonus}, given
a joint strategy $s$, we can view the external incoming edges into
$\mathcal{C}_{j}$ as bonuses to the corresponding nodes. The only
difference in this case is that the value of the bonus instead of
being 1, is the weight of the corresponding edge. Let $\beta^s_j$
denote this bonus function. The cycle $\mathcal{C}_{j}$ along with
$\beta^s_j$ defines a coordination game on a weighted simple cycle
with at most two nodes having non-trivial bonuses. By
Theorem~\ref{thm:cycle-2bonuses}, such a coordination game is weakly
acyclic. 

Let $s^0$ be an arbitrary joint strategy in the game whose underlying
graph is $\mathcal{N}$. We construct a finite improvement path
inductively as follows. Initially, the improvement path consists of
the joint strategy $s^0$. Suppose we have constructed an improvement
path $\rho'$ such that $\lasts(\rho')=s'$. Choose the least $j \in
\{1,\ldots,m\}$ such that there is a node in $\mathcal{C}_j$ which is
not playing its best response in $s'$. Apply
Theorem~\ref{thm:cycle-2bonuses} to the game induced by
$\mathcal{C}_j$ and $\beta_j^{s'}$ to extend the improvement
path. Since $w_1^j > w_2^j$ for all $j$, we can show that the partial
improvement path $\rho$ that is constructed in this manner satisfies
the following invariant:

\begin{itemize}
\item[(I)]Let $\lasts(\rho)=s$ and let $j$ be the largest index $j \in
  \{1, \ldots m-1\}$ such that for all $k \leq j$, $i \in \{1^k,\ldots
  n^k\}$, $s_{i}$ is a best response to $s_{-i}$. If $s_{1^k} \neq
  s_{n^k}$ then $s_{n^k} \not\in C(1^k)$.
\end{itemize}

The invariant asserts that if in the strategy $s$, the choice of the
nodes $1^k$ and its unique predecessor $n^k$ in $\mathcal{C}_k$ are
not the same then the colour chosen by $n^k$ is not in the available
colours for $1^k$.

To see how the above process works, suppose there is a node in
$\mathcal{C}_1$ which is not playing its best response in $s$. Start
with the coordination game induced by $\mathcal{C}_1$ and
$\beta_1^s$. By Theorem~\ref{thm:cycle-2bonuses}, there is a finite
improvement path that terminates in a joint strategy $s^1$ such that
for all $i \in \{1^1,\ldots n^1\}$, $s^1_{i}$ is a best response to
$s^1_{-i}$. Since $w_1^1 > w_2^1$ the invariant (I) holds for the node
$1^1$.

Now suppose we have constructed a improvement path $\rho$ where
$\lasts(\rho)=s$ and let $j$ be the largest index $j \in \{1, \ldots
m-1\}$ such that for all $k \leq j$, $i \in \{1^k,\ldots n^k\}$,
$s_{i}$ is a best response to $s_{-i}$ and the invariant (I)
holds.  Consider the coordination game induced by the
$\mathcal{C}_{j+1}$ and $\beta_{j+1}^{s^j}$. By
Theorem~\ref{thm:cycle-2bonuses}, there is a finite improvement path
that terminates in a joint strategy $s^{j+1}$ such that for all $i \in
\{1^{j+1},\ldots n^{j+1}\}$, $s^{j+1}_{i}$ is a best response to
$s^{j+1}_{-i}$. Now since the two cycles $\mathcal{C}_{j}$ and
$\mathcal{C}_{j+1}$ share a node, i.e.\ $1^j=k^{j+1}$. It is possible
that in the joint strategy $s^{j+1}$, the node $2^j$ is not playing
its best response any longer. To avoid multiple subscripts, let us
denote the node $2^j$ by $2$, $1^j$ by $1$, $n^j$ by $n$ and
$(k-1)^{j+1}$ by $k-1$. So we have that $s^{j+1}_{2}$ is not a best
response to $s^{j+1}_{-2}$. Note that by assumption the node $2^j$ was
playing its best response in the joint strategy $s^j$. And the only
node in $\mathcal{C}_j$ that could possibly change its strategy in
$s^{j+1}$ is $1^j$. If $1^j$ changes its strategy then this means that
$s^{j+1}_1 \neq s^{j+1}_n$. By invariant (I), this means $s^{j+1}_n
\not\in C(1)$ and so $s^{j+1}_1=s^{j+1}_{k-1}$. We also have,
$s^{j+1}_{1} \neq s^{j+1}_{2}$, $s^{j+1}_{1} \in C(s^{j+1}_{2})$ and
$p_{2}(s^{j+1}_{1}, s^{j+1}_{-2}) \geq p_{2}(c, s^{j+1}_{-2})$ for all
$c \in C(2^j)$. Let the node $2^j$ switch to the colour
$s^{j+1}_{1}$. We then update the nodes in the cyclic order in
$\mathcal{C}_j$ successively if they are not playing their best
response. It can be verified that for every node which is not playing
its best response, the colour $s^{j+1}_{1}$ is a best
response. Therefore the only colour which is propagated is
$s^{j+1}_{1}$. So this sequence of updates terminate in a joint
strategy in which all the nodes in $\mathcal{C}_j$ and
$\mathcal{C}_{j+1}$ is playing their best response.

In this resulting joint strategy it could be that the node $2^{j-1}$
is not playing the best response (since the node $1^{j-1}=k^j$
switched). Again by the same reasoning, and by invariant (I), we can
argue that in this case we can update the strategies of the players
such that only the colour $s^{j+1}_{1}$ is propagated. Continuing in
this manner we arrive at a joint strategy in which all nodes on cycles
$\mathcal{C}_1, \ldots, \mathcal{C}_{j+1}$ are playing their best
response. Since $w_1^{j+1} > w_2^{j+1}$ the invariant (I) continues to
hold. In case the weights on the incoming edges are not distinct, then
depending on the initial joint strategy $s^0$ it is possible that a
new colour is introduced when we propagate down the chain
$\mathcal{C}_j,\ldots,\mathcal{C}_1$. In this case we can identify
{\it break points} and use a similar technique as done in the proof of
Theorem~\ref{thm:necklace-noweight-nobonus} to identify a progress
measure.
\end{proof}

\ochainweight*
\begin{proof}
Let $\mathcal{N}$ be the \ochain{} consisting of the sequence of
weighted cycles $\mathcal{C}_1, \ldots, \mathcal{C}_m$.%
to combine
A {\it block} $B_j$ is a sequence of simple cycles in
$\mathcal{N}$. We can represent $\mathcal{N}$ as a sequence of blocks
which we define inductively as follows: The block $B_1$ consists of
the sequence of cycles $\mathcal{C}_1,\ldots,\mathcal{C}_{l}$ such
that for all $k \in \{1,\ldots, l-1\}$, $w_1^k > w_2^k$ or for all $k
\in \{1,\ldots, l\}$, $w_1^k < w_2^k$.

Suppose we have inductively constructed the block $B_j$ and let
$\mathcal{C}_p$ be the last cycle in $B_j$. Then $B_{j+1}$ consists of
the sequence of cycles $\mathcal{C}_{p+1}, \ldots,
\mathcal{C}_{q}$ such that one of the following conditions hold, 
\begin{itemize}
\item for all $k \in \{p+1,\ldots, q-1\}$, $w_1^k > w_2^k$ and if $q
  \neq m$ then $w_1^q < w_2^q$,
\item for all $k \in \{p+1,\ldots, q-1\}$, $w_1^k < w_2^k$ and if $q
  \neq m$ then $w_1^q > w_2^q$,
\end{itemize}

Thus the \ochain{} consisting of the sequence of weighted cycles can
now be represented as a sequence of blocks $B_1, \ldots, B_l$.  The
pair of blocks $B_i$ and $B_{i+1}$ share a node in common. Let $s^0$
be an arbitrary joint strategy in the game whose underlying graph is
$\mathcal{N}$. We construct a finite improvement path inductively as
follows. Initially, the improvement path consists of the joint
strategy $s^0$. Suppose we have constructed an improvement path
$\rho'$ such that $\lasts(\rho')=s'$. Choose the least $j \in
\{1,\ldots,l\}$ such that there is a node in the block $B_j$ which is
not playing its best response in $s'$. Let $B_j$ consists of the
cycles $\mathcal{C}_p,\ldots,\mathcal{C}_{q}$. If for all $k \in \{p,
\ldots, q-1\}$, $w_1^k > w_2^k$ then apply
Lemma~\ref{lm:weighted-chain1} to the sequence
$\mathcal{C}_p,\ldots,\mathcal{C}_{q}$ with the possibility of bonus
to a node in $\mathcal{C}_p$ and $\mathcal{C}_{q}$ to extend the
improvement path. If for all $k \in \{p, \ldots, q-1\}$, $w_1^k <
w_2^k$ then apply Lemma~\ref{lm:weighted-chain2} to extend the
improvement path.

The proof that this procedure constructs a finite improvement path is
similar to the proof of
Theorem~\ref{thm:necklace-noweight-nobonus}. Suppose we have
constructed a partial improvement path $\rho$ where $\lasts(\rho)=s$
and $j$ is the largest index such that all nodes in blocks $B_1,
\ldots, B_j$ are playing the best response in $s$. Consider the block
$B_{j+1}$, by applying either Lemma~\ref{lm:weighted-chain1} or
Lemma~\ref{lm:weighted-chain2} (depending on the case), we can extend
the improvement path to $\rho^1$ such that in
$s^1=\lasts(\rho^1)$. Let us assume that the block $B_{j+1}$ consists
of the cycles $\mathcal{C}_{p+1},\ldots,\mathcal{C}_{q}$. Then $B_{j}$
and $B_{j+1}$ share a common node, $1^p$. If the strategy of the node
$1^p$ in $s$ and $s^1$ is the same, then in $s^1$ we have strictly
increased the number of blocks playing the best response. Suppose
$s^1_{1^p} \neq s_{1^p}$, then there are two cases to analyse. Suppose
in the block $B_j$, for all cycles $\mathcal{C}_k$, $w_1^k >
w_2^k$. Then we can argue that the only colour which is propagated is
$s^1_{1^p}$ and therefore, after applying
Lemma~\ref{lm:weighted-chain1}, the number of blocks playing the best
response increases.  Suppose in the block $B_j$, for all cycles
$\mathcal{C}_k$, $w_1^k < w_2^k$. By the procedure explained in
Lemma~\ref{lm:weighted-chain2} we can reach a joint strategy $s^2$
such that all nodes in $B_j$ is playing their best response. Now if
$s^2_{1^p}=s^1_{1^p}$ then all the nodes in $B_{j+1}$ is also playing
their best response and therefore, the number of blocks playing the
best response increases. Suppose $s^2_{1^p} \neq s^1_{1^p}$ then it
has to be the case that $s^2_{1^p} = s^2_{n^p}$. Note that by
definition of blocks, $w_1^p > w_2^p$. Like in the proof of
Theorem~\ref{thm:necklace-noweight-nobonus} we can define $1^p$ to be
a break-point in $s^2$ since $w_1^p > w_2^p$, $s^2_{1^p} = s^2_{n^p}$
and all nodes in $B_j$ are playing their best response. Similar to the
proof of Theorem~\ref{thm:necklace-noweight-nobonus} we can argue that
after each such phase, either the number of block playing the best
response strictly increases or the value of the maximal break point
strictly increases. 

The improvement path constructed by applying
Lemma~\ref{lm:weighted-chain1} and Lemma~\ref{lm:weighted-chain2} can
be of length $\mathcal{O}(nm^2)$. While composing this path we might
have to propagate colours down the chain. However, we can argue that
we always make progress by at least one block. Thus in worst case, the
length of the improvement path can be $\mathcal{O}(nm^3)$.
\end{proof}

\cchainnoweight*

\begin{proof}
Let $\{\mathcal{C}_j \mid j \in \{1,2, \ldots, m\}\}$ be the set of
simple cycles which constitute the graph $\mathcal{C}$. In each
$\mathcal{C}_j$, there is exactly two nodes with indegree two. By the
definition of $\mathcal{C}$, the simple cycles $\mathcal{C}_1$ and
$\mathcal{C}_m$ share one node $1^m=k^1$ for some $k \in \{1, \ldots
n\}$. Let $s^0$ be an arbitrary joint strategy in the game whose
underlying graph is $\mathcal{C}$. The idea of the proof is the view
the sequence of cycles $\mathcal{C}_1, \ldots, \mathcal{C}_{m-1}$ as
an \ochain{}. By Theorem~\ref{thm:necklace-noweight-nobonus}, there is
a finite improvement path starting at $s^0$ and terminating in $s^1$
such that for all nodes in cycles $\mathcal{C}_1, \ldots,
\mathcal{C}_{m-1}$ are playing their best response in $s^1$. Now we
update the strategies of nodes in $\mathcal{C}_{m}$ in the cyclic
order, let the resulting joint strategy be $s^2$ if the nodes $1^m$
and $1^{m-1}$ choose the same strategy in both $s^1$ and $s^2$ then we
have constructed the finite improvement path.

Suppose $s^2_{1^{m-1}} \neq s^1_{1^{m-1}}$ and $s^2_{1^{m-1}} =
s^2_{(k-1)^{m}}$ (where the nodes $1^{m-1}$ and $k^m$ are the
same). In $s^2$ the node $2^{m-1}$ may no longer be playing the best
response. We proceed in the reverse order and update the nodes in the
cycles $\mathcal{C}_{m-1}, \mathcal{C}_{m-2}, \ldots, \mathcal{C}_{1},
\mathcal{C}_{m}$. If no new colour is introduced and the only colour
which is propagated is $s^2_{1^{m-1}}$ then the improvement path
terminates after updating nodes in $\mathcal{C}_{m}$. If a new colour
$c' \neq s^2_{(k-1)^{m}}$ is introduced then let $1^q$ be the first
time this happens while updating players in the order of cycles
$\mathcal{C}_{m-1}, \mathcal{C}_{m-2}, \ldots, \mathcal{C}_{1}$ (note
that a new colour can be introduced only by a node with indegree
2). Let $s^3$ be the resulting joint strategy, then due to the order
of scheduling nodes, it follows that $s^3_{(k-1)^{q+1}} \not\in
C(1^q)$ and $s^3_{1^q}= s^3_{n^q}$. Each time a new colour is
introduced, for the node involved, the above condition is
satisfies. In other words, the node forms a {\it break point} for that
particular joint strategy as defined in the proof of
Theorem~\ref{thm:necklace-noweight-nobonus}. The important observation
is that, since the new colour of node $1^q$ is supported by the node
$n^q$, the payoff for $1^q$ is at least 1 and therefore while we
update nodes in the reverse order of cycles, if no more new colours
are introduced, then the only colour which is propagated further down
the chain is $s^3_{1^q}$ and then the path terminates at
$(k-1)^{(q-1)}$. Other new colours could be introduced in this
propagation. However, the node which introduces the new colour is then
a break point. Let $1^r$ be the last node where a new colour is
introduced and the resulting joint strategy be $s^4$. This implies
that $s^4_{(k-1)^{r+1}} \not\in C(1^r)$ and $s^3_{1^r}=
s^3_{n^r}$. Now we schedule the cycles $\mathcal{C}_r,
\mathcal{C}_{r-1}, \ldots \mathcal{C}_{q+1}$. The only colour which is
propagated is the colour of $1^r$ and the node $1^q$ does not update
its strategy since the colour chosen by $1^q$ and $n^q$ is the
same. So after this, all nodes in the cycles $\mathcal{C}_r,
\mathcal{C}_{r-1}, \ldots \mathcal{C}_{q+1}$ are on their best
response. Let the resulting joint strategy be $s^5$. It could still be
that $2^r$ is not on its best response (since $1^r$ updated the colour
to a new colour). We now update the nodes in the order of cycles
$\mathcal{C}_{r}, \mathcal{C}_{r+1}, \ldots, \mathcal{C}_{q}$, this
propagates the colour $s^4_{1^r}(=s^5_{1^r})$. If there is a node
$1^l$ such that the colour chosen by $1^l$ is same as that of $n^l$
and $s^4_{1^r} \not\in C(1^l)$ then the propagation stops. If not,
then the same colour $s^4_{1^r}$ is propagated and the improvement
path terminates at the cycle $\mathcal{C}_{r+1}$.

If $s^2_{1^{m-1}} = s^1_{1^{m-1}}$ it could still be that $s^2_{1^{1}}
\neq s^1_{1^1}$. In this case we update the players in the increasing
order of cycles $\mathcal{C}_1, \ldots \mathcal{C}_n$ and using a
similar argument as above, we can show that a finite improvement path
can be constructed.
\end{proof}

\secchain*

\begin{proof}
Let $s$ be a \NE that this game reaches via
an improvement path of length $\mathcal{O}(nm^2)$
as constructed in Theorem \ref{thm:necklace-cycle-noweights}. 
If $s$ is a strong equilibrium then we are done.
Otherwise there exists a coalition $K$ with 
a profitable deviation, $s'$, from $s$.
Due to Lemma \ref{lem:unicolor-cycle}, the coalition $K$ has 
to include at least one simple cycle, $\mathcal{C}$, 
switching to the same colour in $s'$.
This can be one of the cycles $\mathcal{C}_i$ or
one of the two cycles going around the whole game graph 
containing the set of nodes $A = \{1^j | 1 \leq j \leq m\}$.

Note that $s'$ may not be a \NE but 
because the game is weakly acyclic there is a finite improvement path
which leads to a \NE $s''$ from $s'$.
We now show that $s''|_{\mathcal{C}} = s'|_{\mathcal{C}}$.
Let $s^*$ be a strategy profile along the path from $s'$ to $s''$
when for the first time a node, $i$, from $\mathcal{C}$ switches its colour.
We have $p_i(s^*) \leq 1$ %
because
$i$ has at most two incoming edges and one of them is from a node in $\mathcal{C}$.
At the same time, $p_i(s^*) \geq p_i(s') + 1 \geq p_i(s) + 2 \geq 2$, because
the deviation of $i$ to $s^*_i$ is assumed to be profitable and
so is the deviation, as part of coalition $K$, to $s'$; a contradiction.

Now again, either $s''$ is a strong equilibrium, and we are done, or
there exists a new coalition $K'$ with a profitable deviation, $s'''$, from $s''$.
We claim that either $\mathcal{C} = A$ or no node from $\mathcal{C}$ can be part of $K'$.
Any node $i \in \mathcal{C}\setminus A$ has only one incoming edge and so
node $i$ cannot be part of $K'$ and improve any further from $p_i(s'') = 1$. 
Moreover, any successor of $i$ in $\mathcal{C}$ cannot be part of $K'$ either,
because it would need to switch to a different colour than $i$ and 
so cannot improve his payoff of 1 in $s''$.
It follows that either $\mathcal{C} \cap K' = \emptyset$
or $\mathcal{C}\setminus A$ is empty, which implies $\mathcal{C} = A$.
In the latter case, every simple cycle, which has to be part of coalition $K'$, 
has a nonempty intersection with $A$.
Such a node, $i$, would need to improve its payoff to $2$, because $p_i(s'') \geq 1$,
so both of its predecessors have to belong to $K'$. In particular, 
its predecessor in $A$. It follows that $A \subseteq K'$. Furthermore, all predecessors
of nodes in $A$ should belong to $K'$, but this includes all nodes of the game.
Therefore, all nodes in the game have to switch to the same colour which would form a strong equilibrium.
It follows that there can be at most one profitable coalition deviation after coalition $\mathcal{C} = A$ deviates.
So we can safely ignore this special case in the analysis below
and assume that always $\mathcal{C} \cap K' = \emptyset$.

Finally, we construct a finite c-improvement path $\rho = s^{0,0}$,
$s^{1,0}$, $s^{1,1}$, $\ldots, s^{1,k_1}$, $s^{2,0}$, $s^{2,1}$,
$\ldots, s^{2,k_1}$, $s^{3,0}$, $\ldots$ as follows.  It starts with
$s^{0,0} = s$ and we stipulate $k_0 = 0$.  For any $j \geq 1$,
strategy profile $s^{j,0}$ is a result of a profitable deviation by
any coalition from $s^{j-1,k_{j-1}}$.  If there is no such deviation
the path is finished and $s^{j-1,k_{j-1}}$ is a strong equilibrium.
Otherwise, although $s^{j,0}$ may not be a \NE, the game is weakly
acyclic and thanks to Theorem~\ref{thm:necklace-cycle-noweights} there
exist an improvement path $s^{j,1}, s^{j,2}, \ldots, s^{j,k_j}$ of
length $\mathcal{O}(nm^2)$ which reaches a \NE $s^{j,k_j}$.  We know
that in each $s^{j,0}$ at least one simple cycle changes colour and
none of its nodes change colour afterwards.  This shows that the
number of non-unilateral coalition deviation is at most equal to the
number of different simple cycles in the game graph, which is equal to
$m$.  Therefore, $\rho$ is a c-improvement path of length
$\mathcal{O}(nm^3)$.
\end{proof}

\section*{Appendix B --  Partition cycle}

\pcyclenoweight*
\begin{proof}
Consider the initial joint strategy $s^0$. We construct a finite
improvement path starting in $s^0$ by proceeding in the cyclic order
and updating players' strategies. The argument that this results in a
finite improvement path, is very similar to the proof of
Theorem~\ref{thm:partcycle-noweight-nobonus}. The main idea is to
enforce the players to update their strategy based on a specific
priority over colours induced by the bonuses. Let us define
$\mathit{MB}(i)=\{c \in C(i) \mid \text{ for all } c' \in C(i),
\beta(i,c) \geq \beta(i,c')\}$ and $\mathit{Max}(i,s)=\{c \in C(i)
\mid \text{ for all } c' \in C(i), \beta(i,c) + \mathcal{S}(i,c,s)
\geq \beta(i,c') + \mathcal{S}(i,c',s)\}$.

Given a partial improvement path $\rho$ with $\lasts(\rho)=s$, if
$s_i$ is not a best response to $s_{-i}$ and $i \in V_T$, then we
update the strategy of $i$ such that it satisfies the following
property (P2):

\begin{itemize}
\item[(P2)] if $s_{i \ominus 1} \in \mathit{BR}(i, s_{-i})$ and there
  exists a $c \in \mathit{MB}(i)$ such that $p_i(c,s_{-i})=p_i(s_{i
    \ominus 1}, s_{-i})$ then player $i$ switches to $c$ (clearly, in
  this case $c \in \mathit{BR}(i, s_{-i})$ as well).
\end{itemize}

\noindent For $i \in V_B$, we update the strategy of $i$ such that it
satisfies the property (P3):
\begin{itemize}
\item[(P3)] If $s_{i \ominus 1} \in \mathit{BR}(i, s_{-i})$ and there
  exists a $c \in \mathit{Max}(i,s)$ such that $p_i(c,s_{-i})=p_i(s_{i
    \ominus 1}, s_{-i})$ then player $i$ switches to $c$ (clearly, in
  this case $c \in \mathit{BR}(i, s_{-i})$ as well).
\end{itemize}

For all $i \in V_T$, if in an improvement path, player $i$ updates its
strategy then by (P2), it switches to a colour in
$\mathit{MB}(i)$. Due to (P2) and the fact that $i$ has a unique
incoming edge, we can verify that in any subsequent joint strategy
$s^1$, if $i$ updates its strategy to a colour $c'$ then it has to be
that $s^1_{i \ominus 1}=c'$ and $c' \in \mathit{MB}(i)$. Thus we can
assume that after some finite prefix of the improvement path
constructed by updating players in the cyclic ordering, for all nodes
$i \in V_T$, $i$ is choosing a strategy in $\mathit{MB}(i)$.

Let $s^2$ be the resulting joint strategy. Consider the set $X$
constructed at this stage as defined in the proof of
Theorem~\ref{thm:partcycle-noweight-nobonus}. Let $X=\{l_1,\ldots
n\}$. By the construction of $X$, we have that for all $j,k \in X$,
$s^2_j=s^2_k$. By definition of the set $X$, all the nodes $j \in X$
have updated its strategy, and therefore, they conform to property
(P3). For a node $j \in X$, let $s^2_j=c^1$. From (P3), it follows
that for all $j \in \{l_1+1,\ldots n\}$, for all $c^2 \in C(j),
\beta(j,c^1) + \mathcal{S}(j,c^1,s^2)+1 > \beta(j,c^2) +
\mathcal{S}(i,c^2,s^2)$. For node $l_1$, if $s^2_{l_1-1} \neq
s^2_{l_1}$ then $s^2_{l_1} \in \mathit{Max}(l_1,s^2)$ otherwise,
$s^2_{l_1}$ satisfies the same property as above. This implies that in
the next cyclic round of updates, for each node $j\in X$, either $j$
updates to the same strategy as its unique predecessor on the cycle or
$j$ is already on a best response strategy which implies that the
resulting joint strategy is a Nash equilibrium. 
Thus following the argument given in the proof of
Theorem~\ref{thm:partcycle-noweight-nobonus} a finite improvement path
can be constructed.  
\end{proof}

\begin{example}
\label{ex:no-way-to-se}
Consider the coordination game graph in Figure \ref{fig:no-way-to-se}. 
This game graph is strongly connected and in fact there are only three edges missing
to turn it into an undirected graph.
Also, although the game graph is weighted, the weighted edges
can easily be replaced by unweighted ones
just by adding auxiliary nodes without affecting the strong connectedness of the graph.
At the same time, the behaviour of the game on this new unweighted graph 
will essentially be the same as on the original one.
Note that coordination games on undirected unweighted graphs are
known to have FIP \cite{ARSS14}. 
If we do not require strong connectedness of the game graph, 
this example can be slightly simplified by
removing nodes A, B, C and turning bidirectional edges from nodes 4--9 into outgoing edges.

First, let us notice that nodes A, B, and C in this game would never like to switch their colour; 
all of them already have the maximum possible payoff of 5.
This implies that nodes 5 -- 9 will never change their colour either,
e.g. node~5 has at least payoff of 2 for picking $b$
and no matter the colour node~1 chooses,
node~5 will never be better off switching to a different colour.
Therefore, the only nodes that can ever switch colours are nodes 1--3.

Now, let us analyse the initial colouring in Figure \ref{fig:no-way-to-se}.
The payoff of node~2 is 4 and his maximum possible payoff is 5.
However, he can only get payoff 5 if he switches to $c$ and
node~1 switches to $c$. 
The latter is not possible because 
node~1 gets payoff of at most 2 for picking $c$ while picking $a$ gives him
at least 3.
In conclusion, node~2 cannot be part of a deviating coalition in this colouring.
Node~3 will not change his colour either because he gets payoff 3 while the other
colours give him payoff 2.
Therefore, the only node which can switch in any coalition is node~1 and
his only profitable deviation is switching to colour $b$.

Once this switch happens, he gets payoff~4 in the new colouring,
while his maximum payoff is 5.
It can be argued as before that node~1 cannot be part of a deviating coalition
in this new colouring.
However, there are two possible deviating coalitions:
either node 2 unilaterally switches to colour $c$,
or nodes 2 and 3 switch to colour $c$ together.
In the former, the game will be in essentially the same situation 
as with the initial colouring in Figure \ref{fig:no-way-to-se};
one just need to rotate the colours, numbers and the game to the ``left'',
i.e. colour $b$ is $a$ and node $2$ is node $1$ etc.
While in the latter, the games will be in the situation 
essentially the same as in the colouring
encountered after the first switch from the initial colouring.

It is easy to see now that eventually this game arrives at the initial colouring
and the whole process will repeat forever.
Therefore, this game is not weakly acyclic nor c-weakly acyclic.
On the other hand, it has three trivial strong equilibria in which 
all players pick the same colour.
\qed
\end{example}